\documentclass{article}
\usepackage[utf8]{inputenc}
\usepackage{graphicx} 
\usepackage{comment}
\usepackage[margin=1in]{geometry}
\usepackage{amsthm,amsmath,amssymb}
\usepackage[authoryear]{natbib}
\usepackage{caption}
\usepackage{subcaption}
\usepackage{longtable}
\usepackage{tikz}
\usepackage{xcolor}
\usepackage{enumitem}
\usepackage{url}

\newtheorem{theorem}{Theorem}
\newtheorem{lemma}{Lemma}
\newtheorem{corollary}{Corollary}
\newtheorem{proposition}{Proposition}

\newtheorem{assumption}{Assumption}

\newcommand{\indep}{\perp \!\!\! \perp}

\title{On Identification of Optimal Dynamic Treatment Regimes with Proxies of Hidden Confounders}
\author{Jeffrey Zhang and Eric Tchetgen Tchetgen}
\date{}

\begin{document}

\maketitle

\begin{abstract}
We consider identification of optimal dynamic treatment regimes in a setting where time-varying treatments are confounded by hidden time-varying confounders, but 
proxy variables of the unmeasured confounders are available. We show that, with two independent proxy variables at each time point that are sufficiently relevant for the hidden confounders, identification of the joint distribution of counterfactuals is possible, thereby facilitating identification of an optimal treatment regime.

\end{abstract}
\section{Introduction}
\subsection{Background}
Estimating single time point optimal treatment regimes and multiple time point dynamic treatment regimes has received much attention in a diverse array of fields, including biostatistics, computer science, and economics. When data are obtained from a randomized controlled trial, A/B test, or a sequential multiple assignment randomized trial (SMART), the assumptions required for identification and estimation of optimal rules/regimes can be met by design. Alternatively, when one wishes to estimate the optimal regime from observational data, unverifiable assumptions such as the no unmeasured confounding assumption (NUCA) are typically invoked in point exposure settings and its time-varying analog termed sequential randomization assumption (SRA) is likewise invoked in time-varying treatment settings  (\cite{Robins1986AEFFECT}, \cite{Murphy2003OptimalRegimes}). 
In this article, we do not assume that experimental data are available, e.g. SMARTs, nor do we impose SRA. Instead, we allow for confounding by hidden factors for which variables with special causal structure, called proxies, are available, which we leverage to identify an optimal dynamic treatment regime in an observational study.  
Our results are also relevant beyond observational studies, as even when treatment assignment is randomized, there often is noncompliance. In a SMART with noncompliance, randomization can be viewed either as a valid instrumental variable or a treatment confounding proxy, one of the two types of proxies we leverage for identification. 
\subsection{Related Literature}
There are several strands of research related to our work. The proximal causal inference framework, originally introduced in \cite{Miao2018IdentifyingConfounder} for the point exposure scenario, has been extended in many directions. Notably, \cite{Qi2023ProximalConfounding} and \cite{Shen2022OptimalLearning} solve the point exposure optimal treatment regime problem under unmeasured confounding using proxies. Proximal inference tools have also been applied to more complex regimes, such as mediation analysis \citep{Dukes2023ProximalAnalysis}, longitudinal studies \citep{TchetgenTchetgen2024, Ying2023ProximalStudies,Singh2021AEffects}, and off-policy evaluation  \citep{Bennett2021ProximalProcesses,Miao2022Off-PolicyModels, Shi2022}, respectively. These works, similar to ours, require solving nested integral equations. In contrast, none require the challenging task of identifying the joint distribution of counterfactuals as a crucial step towards identifying an optimal dynamic regime. \cite{Ying2023ProximalStudies} only considers models for marginal counterfactual distributions indexed by static regimes. \cite{Bennett2021ProximalProcesses}, \cite{Shi2022}, and \cite{Miao2022Off-PolicyModels} study off-policy evaluation with unmeasured confounding using proxies, but their approaches are not amenable to finding optimal policies/regimes because the nuisance bridge functions they introduce depend on the policy being evaluated. 

Alternatives to proximal causal inference with longitudinal data subject to unmeasured confounding include instrumental variable approaches as well as partial identification/sensitivity analysis approaches. \cite{Michael2023InstrumentalTreatment} study identification of a static regime using instrumental variables. \cite{Chen2023EstimatingVariable} study improvement of dynamic treatment regimes using instrumental variables given partial identification bounds, though they do not formally state how these should be obtained and the assumptions required. \cite{Han2023OptimalOrdering} estimates a partial order of dynamic treatment regimes using an instrumental variable. \cite{Bruns-Smith2023} study robust Q-learning under a model analogous to the marginal sensitivity model of \cite{Tan2006AScores}. 

A recent work that achieves identification of optimal dynamic treatment regimes under unmeasured confounding is \cite{Han2021IdentificationEffects}, using instrumental variables in addition to structural assumptions grounded in econometrics.  \cite{Shahn2022a} estimates optimal regime structural nested mean models under unmeasured confounding using a parallel trends assumption. Besides \cite{Han2021IdentificationEffects} and \cite{Shahn2022a}, to the best of our knowledge, there do not exist general identification results for optimal dynamic treatment regimes in the presence of unmeasured confounding. We fill this gap by leveraging proxy variables under proxy relevance conditions formally known as completeness conditions,  with similar characteristics as those in \cite{Ying2023ProximalStudies}.

Finally, one work that achieves identification of the joint distribution of potential outcomes is that of \cite{Shpitser2023} (see Appendix B), though the authors did not explicitly mention this as their goal was identification of the marginal, nor did they make the connection to optimal treatment regimes. \cite{Shpitser2023} required solutions to similar integral equations as ours for identification, though our proof technique is markedly different. Moreover, beyond identification of the joint distribution of potential outcomes, we 1) provide closed forms of the necessary solutions to the integral equations in the fully discrete setting, 2) provide a generalization to more than 2 timepoints, and 3) in a setting where the downstream outcomes take on a finite number of values, we outline identification using a different type of bridge function.

\section{Preliminaries}
\subsection{Notation}
To simplify the exposition, we consider a two time point setting, though the results can be extended to any finite number of time points as discussed in Section C of the supplementary material. We suppose the data contains $n$ iid draws of a random vector $(U_{0},Y_{0},Z_1,A_{1},W_1,U_{1},Y_{1},Z_2, A_{2}, W_2, Y_{2})$. Here, the $A$ variables are treatments, $U_0$ and $U_1$ represent time-varying unmeasured confounders of the effects of $(A_{1},A_{2})$, $Y_{0}$ is a baseline measurement of the outcome which may also confound the effect of subsequent treatments, $Y_1$ and $Y_2$ are post-treatment outcome measurements, and $Z$ and $W$ are proxies, whose causal structure is introduced in the next section. The overarching goal is to identify the optimal treatment regime based on past treatments and outcomes. In the two time point setting, we consider regimes at time point 2, ($d_2 = d_2(y_1,a_1,y_0)$) that depend on past outcomes, and treatment; and at time point 1 ($d_1 = d_1(y_0)$) that depend on the pretreatment outcome measurement. We note that allowing the $Y$ outcome process to be multivariate poses no additional theoretical difficulty, though the proxy conditions we introduce later would need to apply to these multivariate outcomes. Moreover, in practice, there typically exist measured time-varying confounders. We can straightforwardly account for and utilize such variables, though we cannot identify optimal regimes based upon them, unless proxies of such variables are measured. We focus on the univariate outcome case without measured time-varying confounders solely to simplify the exposition. We formalize the framework in Section 3. We use the summation notation $\sum$, which in a slight abuse of notation, can also be interpreted as integral notation when random variables are continuous. Also, we adopt the convention that uppercase variables denote random variables and lower case denote realizations of a random variable. At times, we write $f(a)$ instead of $f(A=a)$ to represent the density or probability mass function of random variable $A$ evaluated at $a$, but the abbreviation will be clear from context. $E$ denotes the expectation operator. We adopt the $\overline{A}_k$ notation to represent a random variable and all its predecessors, i.e. $\overline{A}_k = (A_1,\ldots,A_k)$ or $\overline{A}_k = (A_0,A_1,\ldots,A_k)$. We use $|A|$ to denote the cardinality of a discrete random variable $A$.
\subsection{Assumptions}
We next introduce formal causal assumptions underlying the data. First, we assume potential outcomes $Y_2(a_1,a_2)$ and $Y_1(a_1)$ are well-defined and there is no interference, such that the former represents a unit's outcome had, possibly contrary to fact, its treatment history been set to $(a_1,a_2)$; the latter is likewise defined. We also make the standard consistency assumption, so that $Y_2 = Y_2(A_1,A_2)$ and $Y_1 = Y_1(A_1)$. Key to the approach is an assumption that the unobserved factors $U_0,U_1$ suffice to account for time-varying confounding, and the absence of relationships between the proxies and other variables, which we formalize as follows. 
\begin{assumption}[Sequential Proximal Latent Randomization Assumption]
\label{assumption: proxy}
The following hold: $\{Y_{2}(a_{1},a_{2}) ,\overline{W}_{2}\} \indep \{A_{2},\overline{Z}_{2}\}|( \overline{U}_{1},\overline{Y}_{1},A_{1}=a_{1})$, $\{{Y}_{1}(a_{1}), Y_{2}(
a_{1},a_{2}) , W_{1}\} \indep \{A_{1},Z_{1}\}|(
U_{0},Y_{0})$, and $\{\overline{W}_2,Y_1\} \indep Z_{1} \mid A_{1}, U_{0},Y_{0}$.
\end{assumption}

The first two conditional independence statements correspond to the assumptions from \cite{Ying2023ProximalStudies}; specifically Equations 2-4. The third assumption is an additional assumption that requires that $Z_1$ has no direct effect on $W_2$. The canonical situation where these conditions might be expected to hold would be when $W_1$ and $W_2$ are negative control outcomes and $Z_1$ and $Z_2$ are negative control exposures. In other words, $A_1$ and $A_2$ do not have direct effects on $W_1$ and $W_2$, respectively and $Z_1$ and $Z_2$ do not have direct effects on $Y_1$ and $Y_2$, respectively. In addition, $Z$ and $W$ should not affect each other. Note that Assumption \ref{assumption: proxy} does not preclude $A_1$ from having an effect on $W_2$. A Directed Acyclic Graph (DAG) that satisfies Assumption \ref{assumption: proxy} is depicted in Figure \ref{fig:dag} below. This DAG is similar to the one depicted in Figure 2 of \cite{Ying2023ProximalStudies}, with the addition of a pre-treatment confounder $Y_0$ and an intermediate outcome $Y_1$. The $Y_0,U_0$ and $Y_1,U_1$ are grouped together in a single node for increased readability, and represents that they have the same parent and child nodes. Previously, \cite{TchetgenTchetgen2024} and \cite{Ying2023ProximalStudies} reanalyzed a study of the effect of an anti-rheumatic therapy called Methotrexate on average number of tender joints at end of follow-up among patients with rheumatoid arthritis. They utilized erythrocyte sedimentation rate as a negative control exposure and patient's global assessment as a negative control outcome. More generally, good candidates for negative control exposures include randomization in a SMART with noncompliance, or a time-varying instrumental variable in a longitudinal observational study. For more examples of both types of proxy variables, we refer the reader to \cite{Shi2020}. 
\begin{figure}
    \centering
    \begin{tikzpicture}[scale = 0.79, node distance=1cm]
    \tikzstyle{var} = [draw, circle, minimum width=0.5cm, minimum height=0.5cm]
    
    
        \node[var] (A) at (0,0) {\scriptsize $Y_0, U_0$};  
        \node[var] (C) at (2,-2) {$Z_1$};  
        \node[var] (D) at (2,0) {$A_1$}; 
        \node[var] (E) at (2,2) {$W_1$};  
        \node[var] (F) at (4,0) {\scriptsize $Y_1, U_1$}; 
        \node[var] (H) at (6,-2) {$Z_2$}; 
        \node[var] (I) at (6,0) {$A_2$};  
        \node[var] (J) at (6,2) {$W_2$}; 
        \node[var] (K) at (8,0) {$Y_2$};

    \draw[->, line width = 0.8mm] (A) to (1,0);

    \draw[->] (C) -- (D);
    \draw[->] (C) -- (H);
    \draw[->] (C) -- (I);
    \draw[->] (D) -- (F);
    \draw[->] (D) -- (H);
    \draw[->] (D) to[out=45, in=135] (I);
    \draw[->] (D) to[out=320, in=220] (K);
    \draw[->] (E) -- (F);
    \draw[->] (E) -- (J);
    \draw[->] (E) -- (K);
    \draw[->] (F) -- (H);
    \draw[->] (F) -- (I);
    \draw[->] (F) -- (J);
    \draw[->] (F) to[out=35, in=145] (K);
    \draw[->] (H) -- (I);
    \draw[->] (I) -- (K);
    \draw[->] (J) -- (K);
    
\end{tikzpicture}
    \caption{An example of a DAG where Assumption \ref{assumption: proxy} holds. The bold arrow emanating from $Y_0$ and $U_0$ represents arrows to every other node.}
    \label{fig:dag}
\end{figure}
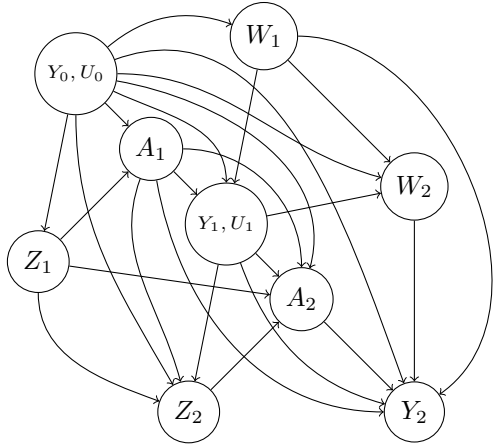

\section{Proximal Identification}
\subsection{Characterizing the optimal regime}
We wish to solve the optimal treatment regime problem, which requires identifying $d_{2}^{opt}(\overline{y}_{1},
\overline{a}_{1})$ and $d_{1}^{opt}(y_{0})$ that solve
\begin{equation}
\label{eq: value function optimal rule}
d_{2}^{opt}(\overline{y}_{1},
\overline{a}_{1}), d_{1}^{opt}(y_{0}) \in \underset{d_1,d_2}{\arg
\max} \ {E}\left\{Y_{2}\left( d_1(Y_0),d_2(Y_1,A_1,Y_0)\right)\right\}.
\end{equation}
In other words, we seek to find the regime at time point 1 based on a baseline outcome $Y_0$ and regime at time point 2 based on a baseline outcome measurement $Y_0$, previous treatment $A_1$, and previous outcome $Y_1$ that will maximize the final outcome $Y_2$. We refer to ${E}\left\{Y_{2}\left( d_1(Y_0),d_2(Y_1,A_1,Y_0)\right)\right\}$ as the value of a regime $(d_1,d_2)$. The regret of an estimated regime is the value of the optimal regime minus the value of the estimated regime. We point out that maximizing the final outcome is without loss of generality, as the intermediate potential outcomes $Y_1(a_1)$ can be used in defining the value function (see Chapter 6.2.1 of \cite{Tsiatis2019DynamicRegimes} for a detailed discussion). In principle, it is possible to find the optimal regime by identifying the value function for any fixed regime and solving the optimization from Equation \ref{eq: value function optimal rule} directly. However, this is typically not computationally feasible, since identification of the value function of a fixed regime using proxies requires estimating nuisance functions that themselves depend on the fixed regime \citep{Bennett2021ProximalProcesses,Shi2022,Miao2022Off-PolicyModels}. This implies that finding the optimal regime directly using Equation \ref{eq: value function optimal rule} would require estimating (on the order of) as many nuisance functions as there are policies, which could be infinite. We turn to the backwards induction approach to circumvent this issue. The backwards induction approach utilizes the well-known equivalent and more convenient characterizations of $d_1^{opt}$ and $d_2^{opt}$ (e.g.  Chapter 7.2 of \cite{Tsiatis2019DynamicRegimes}): 
\begin{equation}
\label{eq: optimal_rules}
\begin{aligned}
d_{2}^{opt}\left(\overline{y}_{1},a_{1}\right) &= \underset{a_{2}}{\arg \max} \ {E}\left\{Y_{2}\left(a_{1},a_{2}\right) |\overline{Y}_{1}\left( a_{1}\right) =\overline{y}_{1}\right\}  \\
V_{2}^{opt}\left( a_{1},\overline{y}_{1}\right)  &=\underset{a_{2}}{\max } \
{E}\left\{Y_2\left( a_{1},a_{2}\right) |\overline{Y}_{1}\left(
a_{1}\right) =\overline{y}_{1}\right\}  \\
d_{1}^{opt}\left( y_{0}\right)  &=\underset{a_{1}}{\arg \max } \ {E}\left\{V_{2}^{opt}\left( a_{1},Y_{1}\left( a_{1}\right) ,Y_{0}\right) |Y_{0}=y_{0}\right\}.
\end{aligned}
\end{equation}
Notice that ${E}
\left\{ V_{2}^{opt}\left( a_{1},Y_{1}\left( a_{1}\right) ,Y_{0}\right)
|Y_{0}=y_{0}\right\} = \sum_{y_1} V_2^{opt}(a_1, y_1, y_0) f\left( Y_{1}\left(
a_{1}\right) =y_{1}|y_{0}\right)$ and ${E}\{ Y_{2}\left( a_{1},a_{2}\right) |\overline{Y}_{1}\left(
a_{1}\right) =\overline{y}_{1}\}  =\sum_{y_{2}}f\left( Y_{2}\left( a_{1},a_{2}\right)
=y_{2},Y_{1}\left( a_{1}\right) =y_{1}|y_{0}\right) /f\left( Y_{1}\left(
a_{1}\right) =y_{1}|y_{0}\right)y_{2}$, 
which clearly indicates that the quantities $f\left( Y_{2}\left( a_{1},a_{2}\right)
=y_{2},Y_{1}\left( a_{1}\right) =y_{1}|y_{0}\right)$ and $f\left( Y_{1}\left(
a_{1}\right) =y_{1}|y_{0}\right)$ are crucial for identification of the optimal regime. The task of identifying optimal regimes with multiple time points is markedly distinct from and more difficult than identification of optimal regimes in the single time point case. This is because identifying the optimal regime using Equation \ref{eq: optimal_rules} requires identification of the joint distribution of the counterfactuals $\{Y_2(a_1,a_2), Y_1(a_1)\}$. Previous work for observational longitidunal studies such as \cite{Michael2023InstrumentalTreatment} and \cite{Ying2023ProximalStudies} provide identification results for functionals of the marginal distribution of the single counterfactual $Y_2(a_1,a_2)$ for some fixed $a_1,a_2$, but not the joint $\{Y_2(a_1,a_2), Y_1(a_1)\}$. By the latent SRA, i.e. Assumption \ref{assumption: proxy}, the joint density of the counterfactuals can be written in terms of observable and the unmeasured confounders by a standard application of Robins' g-formula. Unfortunately, without access to the unmeasured confounders, we must resort to alternative assumptions for identification.
\subsection{Main Result}
Identification of the optimal regime is based on the existence of solutions to certain carefully defined integral equations. Explicitly, we assume the following:
\begin{assumption}[Time-varying Bridge Functions]
\label{assumption: bridge_functs_2}
There exist functions $h_{22}( \overline{y}
_{2},\overline{w}_{2},\overline{a}_{2})$ and $h_{21}( \overline{y}_{1},w_{1},\overline{a}_{2})$ that satisfy
\begin{align*}
&f( Y_{2}=y_{2}|\overline{y}_{1},\overline{u}_{1},\overline{a}%
_{2}) =\sum_{\overline{w}_{2}}h_{22}( \overline{y}_{2},\overline{w%
}_{2},\overline{a}_{2}) f( \overline{w}_{2}|\overline{y}_{1},
\overline{a}_{2},\overline{u}_{1}), \text{ and }\\ &
\sum_{\overline{w}_{2}}h_{22}( \overline{y}_{2},\overline{w}_{2},%
\overline{a}_{2}) f(y_{1}|y_{0},u_{0},a_{1})f( \overline{w}_{2}|%
\overline{y}_{1},a_{1},u_{0}) =\sum_{w_{1}}h_{21}( \overline{y}%
_{2},w_{1},\overline{a}_{2}) f( w_{1}\mid y_{0},a_{1},u_{0}),
\end{align*}
respectively.
\end{assumption}
The integral equation involving the $h_{22}$ function is similar to previous integral equations in the proxy literature (\citealp{Miao2018IdentifyingConfounder,Ying2023ProximalStudies,TchetgenTchetgen2024}), with the left hand side being a conditional density given unmeasured confounders while the right hand side is a conditional mean function. Meanwhile, the integral equation involving $h_{21}$, does not resemble the typical conditional moment integral equations in the proxy literature. Rearranging the integral equation, we get
\begin{equation*}
f(y_{1}|y_{0},u_{0},a_{1}) =\left. \sum_{w_{1}}h_{21}\left( \overline{y}%
_{2},w_{1},\overline{a}_{2}\right) f\left( w_{1}|y_{0},a_{1},u_{0}\right)  \middle/ \sum_{\overline{w}_{2}}h_{22}\left( \overline{y}_{2},\overline{w}_{2},%
\overline{a}_{2}\right) f\left( \overline{w}_{2}|%
\overline{y}_{1},a_{1},u_{0}\right) \right. .
\end{equation*}
This looks somewhat familiar, as the left hand side is a conditional density including the first unmeasured confounder, while the right hand side is the ratio of two conditional expectations, but the numerator is only conditional on $(y_{0},a_{1},u_{0})$ with randomness over $W_1$ while the denominator is conditional on $(y_1,y_{0},a_{1},u_{0})$ with randomness over $(W_1, W_2)$. 
We can now state our main result relating the joint and marginal counterfactual densities to the $h_{21}$ function.
\begin{theorem}
\label{thm: counterfactual_id}
Under Assumptions \ref{assumption: proxy} and \ref{assumption: bridge_functs_2},
\begin{equation*}
f( Y_{2}( a_{1},a_{2}) =y_{2},Y_{1}( a_{1})
=y_{1}|y_{0}) =\sum_{w_{1}}h_{21}( \overline{y}_{2},w_{1},
\overline{a}_{2}) f( w_{1}|y_{0}) \text{ and }
\end{equation*}
\begin{equation*}
f(Y_{1}( a_{1}) =y_{1}|y_{0}) 
=\sum_{y_2}\sum_{w_{1}}h_{21}( \overline{y}_{2},w_{1},
\overline{a}_{2}) f( w_{1}|y_{0}).
\end{equation*}
\end{theorem}
Assumption \ref{assumption: bridge_functs_2} requires the existence of solutions to integral equations involving the unmeasured confounders. Thus, without additional assumptions, the functions $h_{22}$, $h_{21}$ are not identified from the observed data. We next introduce conditions under which identification is possible.
\begin{assumption}[Completeness/Relevance]
\label{assumption: completeness}
For any square integrable $g_1$, we have that 
${E}\{g_1(\overline{U}_1) \mid \overline{y}_1, \overline{z}_2, \overline{a}_2\} = 0 \text{ implies } g_1(\overline{U}_1) =0$. In addition, for any square integrable $g_0$, we have that 
${E}\{g_0({U}_0) \mid {y}_0, {z}_1, {a}_1\} = 0 \text{ implies } g_0({U}_0) =0.$ 
\end{assumption}
Completeness roughly requires that any variability in the unmeasured confounders $U$  induces some variability in $Z$. If all variables are discrete, Assumption \ref{assumption: completeness} requires that $\text{dim}(\Bar{u}_1) \leq \text{dim}(\Bar{z}_2)$ and $\text{dim}({u}_0) \leq \text{dim}({z}_1)$. 

\begin{lemma}
\label{lm: link_latent_to_observed}
Suppose Assumption \ref{assumption: completeness} holds, and
$h_{22}\left( \overline{y}%
_{2},\overline{w}_{2},\overline{a}_{2}\right) $ and $h_{21}\left( \overline{y}_{1},w_{1},\overline{a}_{2}\right) $ solve
\begin{equation}
\label{eq: obs_bridge_functions}
\begin{aligned}
&f\left( Y_{2}=y_{2}|\overline{y}_{1},\overline{z}_{2},\overline{a}%
_{2}\right) =\sum_{\overline{w}_{2}}h_{22}\left( \overline{y}_{2},\overline{w%
}_{2},\overline{a}_{2}\right) f\left( \overline{w}_{2}|\overline{y}_{1},%
\overline{a}_{2},\overline{z}_{2}\right), \\
&\sum_{\overline{w}_{2}}h_{22}\left( \overline{y}_{2},\overline{w}_{2},%
\overline{a}_{2}\right) f(y_{1}|y_{0},z_{1},a_{1})f\left( \overline{w}_{2}|
\overline{y}_{1},a_{1},z_{1}\right) =\sum_{w_{1}}h_{21}\left( \overline{y}%
_{2},w_{1},\overline{a}_{2}\right) f\left( w_{1}|y_{0},a_{1},z_{1}\right).
\end{aligned}
\end{equation}
respectively. Then $h_{22}$ and $h_{21}$ also solve the latent integral equations from Assumption \ref{assumption: bridge_functs_2}.
\end{lemma}
Lemma \ref{lm: link_latent_to_observed} gives a characterization of the bridge functions in terms of observable variables rather than the latent variables as in Assumption \ref{assumption: bridge_functs_2}. Equipped with this result,  identification of the optimal regime is possible if one can find $h_{22}\left( 
\overline{y}_{2},\overline{w}_{2},\overline{a}_{2}\right) $, $h_{21}\left( \overline{y}_{1},w_{1},\overline{a}_{2}\right) $, that respectively satisfy the equations from Equation \ref{eq: obs_bridge_functions}. We note that uniqueness of solutions to the above equations is not necessary for identification of the optimal regime; any solutions suffice. We outline primitive conditions for the existence (and uniqueness) of solutions to integral equations in Section B of the supplementary material. Theorem \ref{thm: counterfactual_id} and Lemma \ref{lm: link_latent_to_observed} establish that the optimal rules/values at time points 2 and 1 are identified using only observed data as described in the following corollary.
\begin{corollary}
\label{cor: optimal_regime}
Under Assumptions \ref{assumption: proxy}, \ref{assumption: completeness}, and given functions $h_{22}$ and $h_{21}$ that satisfy Equation \ref{eq: obs_bridge_functions}, the optimal regime can be characterized as follows:
\begin{equation*}
\begin{aligned}
d_{2}^{opt}\left( \overline{y}_{1},a_{1}\right)  
&=\underset{a_{2}}{\arg \max }\sum_{y_{2}}\frac{\sum_{w_{1}}h_{21}\left( 
\overline{y}_{2},w_{1},\overline{a}_{2}\right) f\left( w_{1}|y_{0}\right) }{
\sum_{y_2}\sum_{w_{1}}h_{21}\left( \overline{y}_{2},w_{1},
\overline{a}_{2}\right) f\left( w_{1}|y_{0}\right)}y_{2}, \\
V_{2}^{opt}\left( a_{1},\overline{y}_{1}\right)  
&=\underset{a_{2}}{\max }\sum_{y_{2}}\frac{\sum_{w_{1}}h_{21}\left( 
\overline{y}_{2},w_{1},\overline{a}_{2}\right) f\left( w_{1}|y_{0}\right) }{\sum_{y_2}\sum_{w_{1}}h_{21}\left( \overline{y}_{2},w_{1},
\overline{a}_{2}\right) f\left( w_{1}|y_{0}\right)}y_{2}, \\
d_{1}^{opt}\left( y_{0}\right) 
&=\underset{a_{1}}{\arg \max }\sum_{y_{1}}\sum_{y_2}\sum_{w_{1}^*}h_{21}\left( \overline{y}_{2},w_{1}^*,
\overline{a}_{2}\right) f\left( w_{1}^*|y_{0}\right) V_{2}^{opt}\left( a_{1},\overline{y}_{1}\right).
\end{aligned}
\end{equation*}
\end{corollary}
In principle, Corollary \ref{cor: optimal_regime} can be used to estimate the optimal dynamic treatment regime. However, in the general case, the bridge functions can be challenging to estimate and compute nonparametrically as they may involve (conditional) densities. In particular, $h_{21}$'s nontrivial dependence on $h_{22}$ is distinct from other bridge functions in the literature (\citealp{Ying2023ProximalStudies, TchetgenTchetgen2024}). In Section \ref{sec: categorical} below, we consider a setting only involving categorical variables as an initial crucial step towards inference in this challenging setting. Inference in the continuous case will be considered elsewhere given the significant challenges it presents.
\subsection{A connection to the single time-point setting}
In the single time-point setting, a single integral equation needs to be solved, rather than nested ones. Specifically, \cite{Miao2018IdentifyingConfounder} require the existence of a function $h_{11}\left( w_{1},
\overline{y}_{1},a_{1}\right) $ that solves the following so-called Fredholm integral equation of the first kind: 
\begin{equation}
\label{eq: bridge_funct_1}
f\left( Y_{1}=y_{1}|y_{0},a_{1},u_{0}\right) =\sum_{w_{1}}h_{11}\left( w_{1},%
\overline{y}_{1},a_{1}\right) f\left( w_{1}|u_{0},y_{0}\right).
\end{equation}
\cite{Miao2018IdentifyingConfounder} show that $f\left( Y_{1}(a_1)=y_{1}|y_{0}\right) =\sum_{w_{1}}h_{11}\left( w_{1},%
\overline{y}_{1},a_{1}\right) f\left( w_{1}|y_{0}\right)$ for any $h_{11}$ that solves Equation \ref{eq: bridge_funct_1}. Thus, in principle, we could solve for such an $h_{11}$ directly. It turns out that $h_{21}$ is related to $h_{11}$. We establish this relationship in the proposition below.
\begin{proposition}
\label{prop: h_21 h_11 relationship}
For any $h_{21}$ that satisfies the conditions in Assumption \ref{assumption: bridge_functs_2}, the function $\sum_{y_2} h_{21}(\overline{y}_2,w_1,\overline{a}_2)$ solves Equation \ref{eq: bridge_funct_1}.
\end{proposition}
From Proposition \ref{prop: h_21 h_11 relationship}, we see that in principle, it does not make a difference whether we marginalize $h_{21}$ over $y_2$ to identify the single counterfactual law $f(Y_1(a_1) = y_1 \mid y_0)$ or identify it using $h_{11}$ directly from Equation \ref{eq: bridge_funct_1}.

\section{Estimation of Bridge Functions}
\subsection{Categorical Variables}
\label{sec: categorical}
In this section, we consider in detail the canonical case where variables $U_{t-1}$, $W_t$ and $Z_t$ are categorical with $| W_t |= | Z_t |  = | U_{t-1}| = k$, so that all proxies and unmeasured confounders have the same cardinality. Assuming equal cardinality is restrictive, however it is particularly instructive as identifying assumptions are easier to interpret and explicit simple closed-form expressions to the integral equations can be obtained which simplifies estimation of the optimal regime. In Section D of the supplementary material, we consider the more general setting where  $\min(| W_t|,| Z_t|)\geq | U_{t-1}|$, which appears to be necessary for nonparametric identification, whereby the cardinality of hidden factors $U$ must not exceed that of proxies. Throughout, we suppose that the treatment at each time point is binary, $| A_t | = 2$. We proceed similarly to \cite{Shi2020MultiplyConfounding} with the following notation: we write $P(W \mid u)=\left\{\operatorname{pr}\left(w_1 \mid u\right), \ldots, \operatorname{pr}\left(w_k \mid u\right)\right\}^{\mathrm{T}}, P(w \mid$ $U)=\left\{\operatorname{pr}\left(w \mid u_1\right), \ldots, \operatorname{pr}\left(w \mid u_k\right)\right\}$ and $P(W \mid U)=\left\{P\left(W \mid u_1\right), \ldots, P\left(W \mid u_k\right)\right\}$ for a column vector, a row vector and a matrix, respectively, that consist of conditional probabilities $\operatorname{pr}(w \mid u)$. For all other variables, vectors and matrices consisting of conditional probabilities are analogously defined. Specifically, $P(U \mid$ $z, x)=\left\{\operatorname{pr}\left(u_1 \mid z, x\right), \ldots, \operatorname{pr}\left(u_k \mid z, x\right)\right\}^{\mathrm{T}}$, $P(U \mid Z, x)=\left\{P\left(U \mid z_1, x\right), \ldots, P\left(U \mid z_k, x\right)\right\}$, and $P(y \mid Z, x)=$ $\left\{\operatorname{pr}\left(y \mid z_1, x\right), \ldots, \operatorname{pr}\left(y \mid z_k, x\right)\right\}$. Interestingly, as pointed out in \cite{Miao2018IdentifyingConfounder}, a key simplification that occurs in the categorical case is that the completeness condition reduces to a rank condition, which under the equal cardinality condition, amounts to invertibility of certain key matrices.
\begin{assumption}[Invertibility]
\label{assumption: categorical_completeness}
For all $(\overline{y}_1,\overline{a}_2)$, the matrices
$P(\overline{W}_2 \mid \overline{U}_1, \overline{y}_1,\overline{a}_2)$ and $P(\overline{U}_1 \mid \overline{Z}_2,\overline{y}_1,\overline{a}_2)$ are invertible. In addition, for all $({y}_0,{a}_1)$, $P({W}_1 \mid {U}_0, {y}_0,{a}_1)$ and $P({U}_1 \mid {Z}_1,{y}_0,{a}_1)$ are invertible.
\end{assumption}
These invertibility assumptions require that the corresponding matrices are full rank. 
We can now state a result about identification of $f(Y_2(a_1,a_2)=y_2,Y_1(a_1)=y_1 \mid y_0)$ and $f(Y_1(a_1)=y_1 \mid y_0)$ in the categorical setting.
\begin{proposition}
\label{prop: discrete-id}
Let the cardinality of each of the $U$, $Z$, and $W$ variables equal $k$ for $2 \leq k < \infty$. Also, suppose that Assumptions \ref{assumption: proxy} and \ref{assumption: categorical_completeness} hold. Then
\begin{equation*}
\begin{aligned}
&f(Y_2(a_1,a_2)=y_2,Y_1(a_1)=y_1 \mid y_0) \\ &= P(y_2 \mid \overline{Z}_2,\overline{y}_1,\overline{a}_2)P(\overline{W}_2 \mid \overline{Z}_2, \overline{y}_1,\overline{a}_2)^{-1}P(\overline{W}_2,y_1 \mid Z_1, y_0, a_1) P(W_1 \mid Z_1, y_0, a_1)^{-1}P(W_1 \mid  y_0), \\ 
&f(Y_1(a_1)=y_1 \mid y_0) = P(y_1 \mid \overline{Z}_1,\overline{y}_0,\overline{a}_1) P(W_1 \mid Z_1, y_0, a_1)^{-1}P(W_1 \mid  y_0).
\end{aligned}
\end{equation*}
\end{proposition}
Given the proceeding developments, simple substitution estimators are immediate. In fact, one may simply replace unknown population probabilities with empirical estimates. Then following the above, one obtains plug-in estimators $\widehat{f}(Y_2(a_1,a_2)=y_2,Y_1(a_1)=y_1 \mid y_0)$ and $\widehat{f}(Y_1(a_1)=y_1 \mid y_0)$. Using these plug-in estimators, one then can find the treatment values maximizing corresponding value function (or minimizing regret) by appealing to Equation \ref{eq: optimal_rules}. In turn, one obtains estimates of the optimal treatment strategy at both time steps for any combination of histories of outcome and treatment variable. Explicit expressions are contained in Section E of the supplementary material. With categorical variables, all (conditional) probabilities can be estimated from the observed cell counts.

\subsection{Categorical Outcomes}
In this subsection, we outline another situation in which it is feasible to estimate solutions to the integral equations. As noted previously, the integral equations described above are extremely difficult to estimate in the general case, since they involve conditional densities. Therefore, we consider a simpler case where $Y_1$ and $Y_2$ are assumed to be discrete, but all other variables are unrestricted. In such a case, the first observed bridge function from Equation \ref{eq: obs_bridge_functions} can be characterized as follows:

\begin{equation}
\begin{aligned}
    E[\mathbf{1}(Y_2 = y_2) \mid \overline{y}_{1},\overline{a}_{2},\overline{z}_{2}]&=\sum_{\overline{w}_{2}}h_{22}\left( \overline{y}_{2},\overline{w%
}_{2},\overline{a}_{2}\right) f\left( \overline{w}_{2}|\overline{y}_{1},%
\overline{a}_{2},\overline{z}_{2}\right) 
\end{aligned}
\end{equation}
The second bridge function can be rewritten as follows:
\begin{equation}
\begin{aligned}
\sum_{\overline{w}_{2}}h_{22}\left( \overline{y}_{2},\overline{w}_{2},%
\overline{a}_{2}\right) f\left( \overline{w}_{2}|\overline{y}%
_{1},a_{1},z_{1}\right) f(y_{1}|y_{0},\
a_{1},z_{1})&=\sum_{w_{1}}h_{21}\left( \overline{y}_{2},w_{1},\overline{a}%
_{2}\right) f\left( w_{1}|y_{0},a_{1},z_{1}\right) \iff \\ 
\sum_{\overline{w}_{2}}h_{22}\left( \overline{y}_{2},\overline{w}_{2},%
\overline{a}_{2}\right) f\left( \overline{w}_{2}, y_1|{y}_{0},a_{1},z_{1}\right) &=\sum_{w_{1}}h_{21}\left( \overline{y}_{2},w_{1},\overline{a}%
_{2}\right) f\left( w_{1}|y_{0},a_{1},z_{1}\right) \iff \\ 
\sum_{\overline{w}_{2}, \Tilde{y}_1}h_{22}\left( \overline{y}_{2},\overline{w}_{2},%
\overline{a}_{2}\right) \mathbf{1}(\Tilde{y_1} = y_1)f\left( \overline{w}_{2}, \Tilde{y}_1|{y}_{0},a_{1},z_{1}\right) &=\sum_{w_{1}}h_{21}\left( \overline{y}_{2},w_{1},\overline{a}%
_{2}\right) f\left( w_{1}|y_{0},a_{1},z_{1}\right) \iff \\ 
E_{\overline{W}_{2}, {Y}_1}[h_{22}\left( \overline{y}_{2},\overline{W}_{2},%
\overline{a}_{2}\right) \mathbf{1}({Y_1} = y_1) \mid {y}_{0},a_{1},z_{1}] &=E_{W_{1}}[h_{21}\left( \overline{y}_{2},W_{1},\overline{a}%
_{2}\right) \mid y_{0},a_{1},z_{1}].
\end{aligned}
\end{equation}
These characterizations imply that for fixed $\overline{a}_2, y_1,y_2$, the corresponding $h_{22}$ and $h_{21}$ bridge function characterizations are conditional moment equations. Min-max estimators for such bridge functions have been proposed and applied in earlier work \citep{Dikkala2020, Ghassami2022}.

\subsection{Categorical Outcomes - Treatment Bridge Function}
When the outcomes $Y_1$ and $Y_2$ are discrete, we can also identify the identify the joint distribution of the discrete potential outcomes using different bridge functions, which we will refer to as treatment bridge functions. Explicitly, we assume the following:
\begin{assumption}[Time-varying Treatment Bridge Functions]
\label{assumption: bridge_functs_treat}
There exist functions $q_{11}\left( {y}%
_{0},{z}_{1},{a}_{1}\right) $ and $q_{12}\left( \overline{y}_{1},\overline{z}_{2},\overline{a}_{2}\right) $ that satisfy 
\[
\frac{1}{P(A_1 = a_1 \mid u_0, y_0)} = \sum_{z_1}q_{11}\left( {y}%
_{0},{z}_{1},{a}_{1}\right) f(z_1\mid a_1, u_0, y_0),
\] and %
\[
\frac{E[q_{11}\left( {y}%
_{0},{Z}_{1},{a}_{1}\right)\mid a_1, u_0, y_0]}{P(A_2 = a_2 \mid a_1, \overline{u}_1, \overline{y}_1)} = \sum_{\overline{z}_2} q_{12}\left( \overline{y}_{1},\overline{z}_{2},\overline{a}_{2}\right) f(\overline{z}_2 \mid \overline{y}_{1},\overline{u}_{1},\overline{a}_{2}),
\]%
respectively.
\end{assumption}

The treatment bridge functions can be used to identify the joint distribution of the potential outcomes, as described in the following theorem:

\begin{theorem}
\label{thm: counterfactual_id_treat}
    Suppose outcomes $Y_2$ and $Y_1$ are discrete. Under Assumptions \ref{assumption: proxy} and \ref{assumption: bridge_functs_treat},
\begin{equation*}
f\left( Y_{2}\left( a_{1},a_{2}\right) =y_{2},Y_{1}\left( a_{1}\right)
=y_{1}|y_{0}\right) =E[q_{12}\left( \overline{y}_{1},\overline{Z}_{2},%
\overline{a}_{2}\right)\mathbf{1}\{Y_2 = y_2,Y_1 = y_1\}\mathbf{1}\{\overline{A}_2 = \overline{a}_2\} \mid y_0]  \text{ and }
\end{equation*}
\begin{equation*}
f\left(Y_{1}\left( a_{1}\right) =y_{1}|y_{0}\right) 
=\sum_{y_2} E[q_{12}\left( \overline{y}_{1},\overline{Z}_{2},%
\overline{a}_{2}\right)\mathbf{1}\{Y_2 = y_2,Y_1 = y_1\}\mathbf{1}\{\overline{A}_2 = \overline{a}_2\} \mid y_0],
\end{equation*}
where $\mathbf{1}$ denotes the indicator function.
\end{theorem}

We now introduce an analogous condition to Assumption \ref{assumption: completeness} that will characterize the treatment bridge functions in terms of the observed data.

\begin{assumption}[W-Completeness/Relevance]
\label{assumption: w-completeness}
For any square integrable $g_1$, we have that 
\begin{equation*}
    {E}[g_1(\overline{U}_1) \mid \overline{y}_1, \overline{w}_2, \overline{a}_2] = 0 \text{ implies } g_1(\overline{U}_1) =0.
\end{equation*}
In addition, for any square integrable $g_0$, we have that 
\begin{equation*}
    {E}[g_0({U}_0) \mid {y}_0, {w}_1, {a}_1] = 0 \text{ implies } g_0({U}_0) =0. 
\end{equation*}
\end{assumption}

\begin{lemma}
\label{lm: link_latent_to_observed_treatment}
Suppose there exist
$q_{12}\left( \overline{y}%
_{2},\overline{z}_{2},\overline{a}_{2}\right) $ and $q_{11}\left( \overline{y}_{1},z_{1},\overline{a}_{2}\right) $ that satisfy 
\begin{equation}
\label{eq: obs_treat_bridge_functions}
\begin{aligned}
\frac{1}{P(A_1 = a_1 \mid w_1, y_0)} = \sum_{z_1}q_{11}\left( {y}%
_{0},{z}_{1},{a}_{1}\right) f(z_1\mid a_1, w_1, y_0,) \\
\frac{E[q_{11}\left( {y}%
_{0},{Z}_{1},{a}_{1}\right)\mid a_1, w_1, y_0]}{P(A_2 = a_2 \mid a_1, \overline{w}_2, \overline{y}_1)} = \sum_{\overline{z}_2} q_{12}\left( \overline{y}_{1},\overline{z}_{2},\overline{a}_{2}\right) f(\overline{z}_2 \mid \overline{y}_{1},\overline{w}_{2},\overline{a}_{2}),
\end{aligned}
\end{equation}
respectively. If Assumption \ref{assumption: w-completeness} also holds, then $q_{11}$ and $q_{12}$ also solve the latent integral equations from Assumption \ref{assumption: bridge_functs_treat}.
\end{lemma}
Similar to Lemma \ref{lm: link_latent_to_observed}, Lemma
\ref{lm: link_latent_to_observed_treatment} gives a characterization of the treatment bridge functions in terms of observable variables rather than the latent variables. The equations in (\ref{eq: obs_treat_bridge_functions}) are conditional moment equations, and so the min-max estimators from \cite{Dikkala2020} can be readily used to estimate the treatment bridge functions.

\section{Simulation results}
In this section, we investigate the finite sample performance of the closed-form estimator given in the previous section in the setting where all variables are categorical . In a simulation study, we compare the plug-in proxy method to a naive approach that relies on the no unmeasured confounding assumption (NUCA), and an oracle procedure that has access to the true unmeasured confounder. More details about these methods can be found in Section E of the supplementary material. We simulate 11 binary variables according to the DAG in Figure \ref{fig:dag} that satisfies the conditional independence assumptions laid out in Section 1 from a saturated model. Exact details about the data generating mechanism are available in Section E of the supplementary material. The total number of iterations was 1000. We used saturated models and non-saturated log-linear models to fit the joint distribution of the observed data, from which all of the plug-in estimators can be derived respectively. For the log-linear models, we simply fit with higher order interactions up to some constant $K$, to deal with the possibility of sparse or empty cells. The observed data consists of 9 binary variables, so the fully saturated model corresponds to fitting a model with up to order 9 interactions. We compare NUCA with oracle, fully saturated proxy, and proxy with log-linear fit with 4,5,6,7,8-order interactions. 
First, in Table \ref{tab: summary_results}, we display, for each of the 4 sample sizes $N \in \{25000, 50000, 100000, 250000\}$ and the methods oracle, NUCA, proxy, and proxy-LL-6, percentiles $q$ of the regrets (the optimal regime value minus the estimated regime value). The ``proxy" method uses the fully saturated model and the ``proxy-LL-6" method uses the log-linear model with order 6 interactions to fit the joint distribution of the observed data. Here, $q \in \{10,25,50,75,90\}$.  More detailed results and results for the 4,5,7, and 8 way interactions are in the supplementary material. Since all data are binary, and the treatment decision at the first timepoint is based on one variable and the second based on three, there are a finite number of unique decision rules. Moreover,  there exist equivalence classes of regimes that induce the same exact value. This is because rules that are identical at the first time step will lead to only a subset of histories $(Y_0,A_1,Y_1)$, and so rules identical at the first time step need only give the same decision on a subset of histories at time step 2 to yield the same value. The $\epsilon$ in Table \ref{tab: summary_results} correspond to zero regret up to the machine precision in \texttt{R}, roughly $\epsilon =$ 2.220446e-16.

\begin{table}
\centering
\begin{tabular}{lccccccc}
  N & method & $q = 10$ & $q = 25$ & $q = 50$ & $q = 75$ & $q = 90$ & mean \\ \hline
  25,000 & proxy & $< \epsilon$ & $< \epsilon$ & 0.01503 & 0.05214 & 0.09824 & 0.03705 \\ 
   & NUCA & $< \epsilon$ & $< \epsilon$ & 0.09159 & 0.09159 & 0.09159 & 0.06420 \\ 
   & oracle & $< \epsilon$ & $< \epsilon$ & $< \epsilon$ & $< \epsilon$ & $< \epsilon$ & 0.00001 \\ 
   & proxy-LL-6 & $< \epsilon$ & $< \epsilon$ & 0.01503 & 0.04746 & 0.09824 & 0.03574 \\ \hline
  50,000 & proxy & $< \epsilon$ & $< \epsilon$ & 0.00932 & 0.03711 & 0.09159 & 0.02588 \\ 
   & NUCA & $< \epsilon$ & 0.09159 & 0.09159 & 0.09159 & 0.09159 & 0.07052 \\ 
   & oracle & $< \epsilon$ & $< \epsilon$ & $< \epsilon$ & $< \epsilon$ & $< \epsilon$ & $< \epsilon$ \\ 
   & proxy-LL-6 & $< \epsilon$ & $< \epsilon$ & $< \epsilon$ & 0.03711 & 0.09159 & 0.02626 \\ \hline
  100,000 & proxy & $< \epsilon$ & $< \epsilon$ & $< \epsilon$ & 0.01503 & 0.09159 & 0.01894 \\ 
   & NUCA & $< \epsilon$ & 0.09159 & 0.09159 & 0.09159 & 0.09159 & 0.07922 \\ 
   & oracle & $< \epsilon$ & $< \epsilon$ & $< \epsilon$ & $< \epsilon$ & $< \epsilon$ & $< \epsilon$ \\ 
   & proxy-LL-6 & $< \epsilon$ & $< \epsilon$ & $< \epsilon$ & 0.01503 & 0.09159 & 0.01927 \\ \hline
  250,000 & proxy & $< \epsilon$ & $< \epsilon$ & $< \epsilon$ & 0.01503 & 0.06117 & 0.01521 \\ 
   & NUCA & 0.09159 & 0.09159 & 0.09159 & 0.09159 & 0.09159 & 0.08756 \\ 
   & oracle & $< \epsilon$ & $< \epsilon$ & $< \epsilon$ & $< \epsilon$ & $< \epsilon$ & $< \epsilon$ \\ 
   & proxy-LL-6 & $< \epsilon$ & $< \epsilon$ & $< \epsilon$ & 0.01503 & 0.04613 & 0.01364 \\ 
\end{tabular}
\caption{A summary of regret for each method at varying sample sizes over 1000 iterations. $q$ denotes the $q$th percentile of regret.}
\label{tab: summary_results}
\end{table}

We notice that the oracle method, at all sample sizes, consistently identifies the optimal regime. Meanwhile, over 1000 iterations, the proxy methods outperform NUCA in terms of the average regret and $q$-quantile regret for all but $q = 90$ for sample size $25,000$, where the proxy methods exhibit more variability. The gaps widen as the sample size increases. For the first three sample sizes, the completely saturated proxy estimator performs similarly to the order 6 interaction proxy estimator, but the former has worse worst-case performance than the latter at $N = 250,000$. 

\section{Discussion}
The current literature on estimating optimal dynamic treatment regimes in the presence of unmeasured confounding is sparse. Although we establish identification when proxy variables are measured and completeness conditions are met, there are many promising avenues for future research. Estimation and inference in the fully general case where cardinalities of variables are unrestricted can be significantly more challenging, as it potentially involves solving ill-posed integral equations. It would also be interesting to develop new highly efficient and robust semiparametric estimators that exhibit better worst-case performance than the simple plug-in estimators we examined in the simulation study. Finally, examining the implications of non-unique solutions to the integral equations (\cite{Zhang2023}, \cite{Bennett2023}) could be of interest.

\section*{Acknowledgement}
The authors gratefully acknowledge support from the National Institutes of Health.

\bibliographystyle{abbrvnat}
\bibliography{references}

\appendix
\renewcommand{\thesection}{\Alph{section}}
\section{Section A: Proofs of main results}
\label{section a: proofs}
\subsection*{Known results}
We first state some well-known results about identification in longitudinal observational studies under sequential randomization. First, recall that the joint potential outcome density can be written as follows:
\begin{equation}
\begin{aligned}
&f\left( Y_{2}\left( a_{1},a_{2}\right) =y_{2},Y_{1}\left( a_{1}\right)
=y_{1}|y_{0}\right)  \\
&=\sum_{u_{0}}f\left(Y_{2}\left( a_{1},a_{2}\right) =y_{2},Y_{1}\left(
a_{1}\right) =y_{1}|y_{0},u_{0}\right) f(u_{0}|y_{0}) \\
&=\sum_{u_{0}}f\left(Y_{2}\left( a_{1},a_{2}\right) =y_{2},Y_{1}\left(
a_{1}\right) =y_{1}|y_{0},u_{0},a_1\right) f(u_{0}|y_{0}) \\
&=\sum_{u_{0}}f\left( Y_{2}\left(a_1, a_{2}\right)
=y_{2},Y_{1}=y_{1}|y_{0},u_{0},a_{1}\right) f(u_{0}|y_{0}) \\
&=\sum_{u_{0},u_{1}}f\left( Y_{2}\left(a_1, a_{2}\right) =y_{2}|\overline{y}%
_{1},\overline{u}_{1},a_{1}\right) f\left( U_{1}=u_{1}|y_{1},y_{0}\
,u_{0},a_{1}\right) f(y_{1}|y_{0},u_{0},a_{1})f(u_{0}|y_{0})
\\
&=\sum_{u_{0},u_{1}}f\left( Y_{2}\left(a_1, a_{2}\right) =y_{2}|\overline{y}%
_{1},\overline{u}_{1},\overline{a}_{2}\right) f\left( U_{1}=u_{1}|y_{1},y_{0}\
,u_{0},a_{1}\right) f(y_{1}|y_{0},u_{0},a_{1})f(u_{0}|y_{0}) \\
&=\sum_{u_{0},u_{1}}f\left( Y_{2}=y_{2}|\overline{y}_{1},\overline{u}_{1},%
\overline{a}_{2}\right) f\left( U_{1}=u_{1}|y_{1},y_{0}\
,u_{0},a_{1}\right) f(y_{1}|y_{0},u_{0},a_{1})f(u_{0}|y_{0}).
\end{aligned}
\end{equation}
The second equality is due to $\{Y_2(a_1,a_2), Y_1(a_1)\} \indep A_1 \mid Y_0, U_0$ and the fifth is due to $Y_2(a_1,a_2) \indep A_2 \mid \overline{Y}_1, \overline{U}_1, A_1 = a_1$. The third and sixth equality are by consistency. Meanwhile, the single counterfactual density is likewise identified as such: 
\begin{equation}
\begin{aligned}
f\left( Y_{1}\left( a_{1}\right) =y_{1}|y_{0}\right)  &=\sum_{u_{0}}f\left(
Y_{1}\left( a_{1}\right) =y_{1}|y_{0},u_{0}\right) f\left(
u_{0}|y_{0}\right)  \\ &=\sum_{u_{0}}f\left(
Y_{1}\left( a_{1}\right) =y_{1}|a_1, y_{0},u_{0}\right) f\left(
u_{0}|y_{0}\right) \\ 
&=\sum_{u_{0}}f\left( Y_{1}=y_{1}|a_{1},y_{0},u_{0}\right) f\left(
u_{0}|y_{0}\right).
\end{aligned}
\end{equation}
The second equality is by $Y_1(a_1) \indep A_1 \mid Y_0, U_0$ and the third is by consistency.
\subsection*{Proof of Theorem \ref{thm: counterfactual_id}}
\begin{proof}
We first show the identification of the joint density.
\begin{eqnarray*}
&&\sum_{w_{1}}h_{21}\left( \overline{y}_{2},w_{1},\overline{a}_{2}\right)
f\left( w_{1}|y_{0}\right)  \\
&=&\sum_{u_{0}}\sum_{w_{1}}h_{21}\left( \overline{y}_{2},w_{1},\overline{a}%
_{2}\right) f\left( w_{1}|y_{0},u_{0}\right) f\left( u_{0}|y_{0}\right)  \\
&=&\sum_{u_{0}}\sum_{w_{1}}h_{21}\left( \overline{y}_{2},w_{1},\overline{a}%
_{2}\right) f\left( w_{1}|y_{0},a_{1},u_{0}\right) f\left(
u_{0}|y_{0}\right)  \\
&=&\sum_{u_{0}}\sum_{\overline{w}_{2}}h_{22}\left( \overline{y}_{2},%
\overline{w}_{2},\overline{a}_{2}\right) f(y_{1}|y_{0},u_{0},a_{1})f\left( 
\overline{w}_{2}|\overline{y}_{1},a_{1},u_{0}\right) f\left(
u_{0}|y_{0}\right)  \\
&=&\sum_{u_{0}}\sum_{\overline{w}_{2}}\sum_{u_{1}}h_{22}\left( \overline{y}%
_{2},\overline{w}_{2},\overline{a}_{2}\right)
f(y_{1}|y_{0},u_{0},a_{1})f\left( \overline{w}_{2}|\overline{y}%
_{1},a_{1},u_{0},u_{1}\right) f\left( u_{1}|\overline{y}_{1},a_{1},u_{0}%
\right) f\left( u_{0}|y_{0}\right)  \\
&=&\sum_{u_{0}}\sum_{u_{1}}\left\{ \sum_{\overline{w}_{2}}h_{22}\left( 
\overline{y}_{2},\overline{w}_{2},\overline{a}_{2}\right) f\left( \overline{w%
}_{2}|\overline{y}_{1},\overline{a}_{2},u_{0},u_{1}\right) \right\} f\left(
u_{1}|\overline{y}_{1},a_{1},u_{0}\right) f\left(
y_{1}|y_{0},a_{1},u_{0}\right) f\left( u_{0}|y_{0}\right)  \\
&=&\sum_{u_{0}}\sum_{u_{1}}f\left( Y_{2}=y_{2}|\overline{y}_{1},\overline{u}%
_{1},\overline{a}_{2}\right) f\left(
u_{1}|\overline{y}_{1},a_{1},u_{0}\right)f\left( y_{1}|y_{0},a_{1},u_{0}\right) f\left(
u_{0}|y_{0}\right)  \\
&=&f\left( Y_{2}\left( a_{1},a_{2}\right) =y_{2},Y_{1}\left( a_{1}\right)
=y_{1}|y_{0}\right) 
\end{eqnarray*}%
The first and fourth equalities follow from basic probability algebra, the second and fifth from Assumption \ref{assumption: proxy} ($A_1 \indep W_1 \mid U_0, Y_0$ and $A_2 \indep \overline{W}_2 \mid \overline{U}_1, \overline{Y}_1, A_1$), the third and sixth from Assumption \ref{assumption: bridge_functs_2}, and the final equality is by the g-formula, restated in the previous subsection.
\end{proof}

\subsection*{Proof of Lemma \ref{lm: link_latent_to_observed}}
The conditions from Assumption \ref{assumption: completeness} were assumed in \cite{Ying2023ProximalStudies}. Using these conditions, showing that if $h_{22}$ solves its respective integral equation on the observed data (i.e. conditioned on $Z$), that it solves the respective integral equation on the latent scale (i.e. conditioned on $U$), follows the usual steps, see \cite{Ying2023ProximalStudies}. The steps for $h_{21}$ are slightly different. 
\begin{proof}
Suppose we have a $h_{21}\left( \overline{y}_{1},w_{1},\overline{a}_{2}\right) $ that satisfies  
\begin{equation}
\sum_{\overline{w}_{2}}h_{22}\left( \overline{y}_{2},\overline{w}_{2},%
\overline{a}_{2}\right) f\left( \overline{w}_{2}|\overline{y}%
_{1},a_{1},z_{1}\right) f(y_{1}|y_{0},\
a_{1},z_{1})=\sum_{w_{1}}h_{21}\left( \overline{y}_{2},w_{1},\overline{a}%
_{2}\right) f\left( w_{1}|y_{0},a_{1},z_{1}\right).
\end{equation}
We can rewrite 
\begin{equation*}
\begin{aligned}
f\left( \overline{w}_{2}|\overline{y}%
_{1},a_{1},z_{1}\right) f(y_{1}|y_{0},\
a_{1},z_{1}) &= f\left( \overline{w}_{2},y_1|{y}%
_{0},a_{1},z_{1}\right) \\&= \sum_{u_0}f\left( \overline{w}_{2},y_1|{y}%
_{0},u_0,a_{1},z_{1}\right)f(u_0\mid{y}%
_{0},a_{1},z_{1}) \\&= \sum_{u_0}f\left( \overline{w}_{2},y_1|{y}%
_{0},u_0,a_{1}\right)f(u_0\mid{y}%
_{0},a_{1},z_{1}),
\end{aligned}
\end{equation*}
where the last equality is by Assumption \ref{assumption: proxy}, namely $\{\overline{W}_2, Y_1\} \indep Z_1 \mid A_1, Y_0, U_0$, and 
\begin{equation*}
\begin{aligned}
\sum_{w_{1}}h_{21}\left( \overline{y}_{2},w_{1},\overline{a}%
_{2}\right) f\left( w_{1}|y_{0},a_{1},z_{1}\right) &= \sum_{w_{1}}h_{21}\left( \overline{y}_{2},w_{1},\overline{a}%
_{2}\right) \sum_{u_0}f\left( w_{1}|y_{0},a_{1},u_0, z_1\right) f(u_0\mid y_{0},a_{1},z_{1})  \\ &=\sum_{w_{1}}h_{21}\left( \overline{y}_{2},w_{1},\overline{a}%
_{2}\right) \sum_{u_0}f\left( w_{1}|y_{0},a_{1},u_0\right) f(u_0\mid y_{0},a_{1},z_{1}), 
\end{aligned}
\end{equation*}
where the second equality is by Assumption \ref{assumption: proxy}, namely ${W}_1 \indep Z_1 \mid A_1, Y_0, U_0$. Then we can write
\begin{equation}
\begin{aligned}
\sum_{\overline{w}_{2}}h_{22}\left( \overline{y}_{2},\overline{w}_{2},%
\overline{a}_{2}\right) &\sum_{u_0}f\left( \overline{w}_{2},y_1|{y}%
_{0},u_0,a_{1}\right)f(u_0\mid{y}%
_{0},a_{1},z_{1})\\&=\sum_{w_{1}}h_{21}\left( \overline{y}_{2},w_{1},\overline{a}%
_{2}\right) \sum_{u_0}f\left( w_{1}|y_{0},a_{1},u_0\right) f(u_0\mid y_{0},a_{1},z_{1}).
\end{aligned}
\end{equation}
Fixing $\overline{y}_2,\overline{a}_2$ and interchanging the order of the integrals, we have
\begin{equation}
\begin{aligned}
\sum_{u_0}\sum_{\overline{w}_{2}}h_{22}\left( \overline{y}_{2},\overline{w}_{2},%
\overline{a}_{2}\right) &f\left( \overline{w}_{2},y_1|{y}%
_{0},u_0,a_{1}\right)f(u_0\mid{y}%
_{0},a_{1},z_{1})\\&=\sum_{u_0}\sum_{w_{1}}h_{21}\left( \overline{y}_{2},w_{1},\overline{a}%
_{2}\right) f\left( w_{1}|y_{0},a_{1},u_0\right) f(u_0\mid y_{0},a_{1},z_{1}).
\end{aligned}
\end{equation}
Then by the second completeness condition in Assumption \ref{assumption: completeness}, 
$$\sum_{\overline{w}_{2}}h_{22}\left( \overline{y}_{2},\overline{w}_{2},%
\overline{a}_{2}\right) f\left( \overline{w}_{2},y_1|{y}%
_{0},u_0,a_{1}\right)= \sum_{w_{1}}h_{21}\left( \overline{y}_{2},w_{1},\overline{a}%
_{2}\right) f\left( w_{1}|y_{0},a_{1},u_0\right).$$ We show the steps for the $h_{22}$ function for sake of completeness. Suppose there is a $h_{22}\left( \overline{y}%
_{2},\overline{w}_{2},\overline{a}_{2}\right)$ that solves
\begin{equation*}
f\left( Y_{2}=y_{2}|\overline{y}_{1},\overline{z}_{2},\overline{a}%
_{2}\right) =\sum_{\overline{w}_{2}}h_{22}\left( \overline{y}_{2},\overline{w%
}_{2},\overline{a}_{2}\right) f\left( \overline{w}_{2}|\overline{y}_{1},%
\overline{a}_{2},\overline{z}_{2}\right).
\end{equation*}

From Assumption \ref{assumption: proxy}, $Y_2 \indep \overline{Z}_2 \mid \overline{Y}_1, \overline{U}_1, \overline{A}_2$ and $\overline{W}_2 \indep \overline{Z}_2 \mid \overline{Y}_1, \overline{U}_1, \overline{A}_2$ hold. Thus, 
\begin{equation*}
\begin{aligned}
f\left( Y_{2}=y_{2}|\overline{y}_{1},\overline{z}_{2},\overline{a}%
_{2}\right) &= \sum_{\overline{u}_1}  f\left( Y_{2}=y_{2}|\overline{y}_{1},\overline{z}_{2},\overline{a}%
_{2}, \overline{u}_1\right) f(\overline{u}_1 \mid \overline{y}_{1},\overline{z}_{2},\overline{a}%
_{2}) \\ &= \sum_{\overline{u}_1}  f\left( Y_{2}=y_{2}|\overline{y}_{1},\overline{a}%
_{2}, \overline{u}_1\right) f(\overline{u}_1 \mid \overline{y}_{1},\overline{z}_{2},\overline{a}%
_{2}), \\ \sum_{\overline{w}_{2}}h_{22}\left( \overline{y}_{2},\overline{w%
}_{2},\overline{a}_{2}\right) f\left( \overline{w}_{2}|\overline{y}_{1},%
\overline{a}_{2},\overline{z}_{2}\right) &= \sum_{\overline{w}_{2}}\sum_{\overline{u}_1}h_{22}\left( \overline{y}_{2},\overline{w%
}_{2},\overline{a}_{2}\right) f\left( \overline{w}_{2}|\overline{y}_{1},%
\overline{a}_{2},\overline{z}_{2}, \overline{u}_1\right) f(\overline{u}_1 \mid \overline{y}_{1},\overline{z}_{2},\overline{a}%
_{2}) \\ &= \sum_{\overline{u}_1}\sum_{\overline{w}_{2}}h_{22}\left( \overline{y}_{2},\overline{w%
}_{2},\overline{a}_{2}\right) f\left( \overline{w}_{2}|\overline{y}_{1},%
\overline{a}_{2}, \overline{u}_1\right) f(\overline{u}_1 \mid \overline{y}_{1},\overline{z}_{2},\overline{a}%
_{2}).
\end{aligned}
\end{equation*}
Applying the first completeness condition from Assumption \ref{assumption: completeness} implies that $$f\left( Y_{2}=y_{2}|\overline{y}_{1},\overline{a}%
_{2}, \overline{u}_1\right) = \sum_{\overline{w}_{2}}h_{22}\left( \overline{y}_{2},\overline{w%
}_{2},\overline{a}_{2}\right) f\left( \overline{w}_{2}|\overline{y}_{1},%
\overline{a}_{2}, \overline{u}_1\right).$$
\end{proof}
\subsection*{Proof of Proposition \ref{prop: h_21 h_11 relationship}}
\begin{proof}
Recall that 
\begin{equation*}
f\left( Y_{2}=y_{2}|\overline{y}_{1},\overline{u}_{1},\overline{a}
_{2}\right) =\sum_{\overline{w}_{2}}h_{22}\left( \overline{y}_{2},\overline{w%
}_{2},\overline{a}_{2}\right) f\left( \overline{w}_{2}|\overline{y}_{1},
\overline{a}_{2},\overline{u}_{1}\right)
\end{equation*}
By causal consistency and Assumption \ref{assumption: proxy}, the left hand side can be rewritten as $f\left(Y_{2}(a_1,a_2)=y_{2}|\overline{y}_{1},\overline{u}_{1},a_1\right)$ and the right hand side as $\sum_{\overline{w}_{2}}h_{22}\left( \overline{y}_{2},\overline{w%
}_{2},\overline{a}_{2}\right) f\left( \overline{w}_{2}|\overline{y}_{1},
{a}_{1},\overline{u}_{1}\right)$. Thus, we get 
\begin{equation*}
\begin{aligned}
&\sum_{u_1}f\left(Y_{2}(a_1,a_2)=y_{2}|\overline{y}_{1},\overline{u}_{1},a_1\right)f(u_1 \mid \overline{y}_{1},{u}_{0},a_1) = \sum_{u_1} \sum_{\overline{w}_{2}}h_{22}\left( \overline{y}_{2},\overline{w%
}_{2},\overline{a}_{2}\right) f\left( \overline{w}_{2}|\overline{y}_{1},
{a}_{1},\overline{u}_{1}\right)f(u_1 \mid \overline{y}_{1},{u}_{0},a_1) \\ 
&\sum_{u_1}f\left(Y_{2}(a_1,a_2)=y_{2}, u_1|\overline{y}_{1},{u}_{0},a_1\right) =  \sum_{\overline{w}_{2}} \sum_{u_1} h_{22}\left( \overline{y}_{2},\overline{w%
}_{2},\overline{a}_{2}\right) f\left( \overline{w}_{2}, u_1|\overline{y}_{1},{u}_{0},a_1\right) \\ 
&f\left(Y_{2}(a_1,a_2)=y_{2}|\overline{y}_{1},{u}_{0},a_1\right) =  \sum_{\overline{w}_{2}} h_{22}\left( \overline{y}_{2},\overline{w%
}_{2},\overline{a}_{2}\right) f\left( \overline{w}_{2}|\overline{y}_{1},{u}_{0},a_1\right) \\ &\implies \sum_{y_2} \sum_{\overline{w}_{2}} h_{22}\left( \overline{y}_{2},\overline{w%
}_{2},\overline{a}_{2}\right) f\left( \overline{w}_{2}|\overline{y}_{1},{u}_{0},a_1\right) = 1
\end{aligned}
\end{equation*}
Plugging this into the integral equation involving $h_{21}$, we get 
\begin{equation*}
\begin{aligned}
\sum_{y_2}\sum_{w_{1}}h_{21}\left( \overline{y}
_{2},w_{1},\overline{a}_{2}\right) f\left( w_{1}|y_{0},a_{1},u_{0}\right)  &= \sum_{y_2} \sum_{\overline{w}_{2}}h_{22}\left( \overline{y}_{2},\overline{w}_{2},%
\overline{a}_{2}\right) f(y_{1}|y_{0},u_{0},a_{1})f\left( \overline{w}_{2}|%
\overline{y}_{1},a_{1},u_{0}\right) \\ &= f(y_{1}|y_{0},u_{0},a_{1})\sum_{y_2} \sum_{\overline{w}_{2}}h_{22}\left( \overline{y}_{2},\overline{w}_{2},%
\overline{a}_{2}\right) f\left( \overline{w}_{2}|%
\overline{y}_{1},a_{1},u_{0}\right) \\ &= f(y_{1}|y_{0},u_{0},a_{1})
\end{aligned}
\end{equation*}
After interchanging summation of $y_2$ and $w_1$, it immediately follows that $\sum_{y_2} h_{21}(\overline{y}_2,w_1,\overline{a}_2)$ solves Equation \ref{eq: bridge_funct_1}.
\end{proof}
\subsection*{Proof of Proposition \ref{prop: discrete-id}}
\begin{proof}
See the proof of Proposition 3 in Section D of the supplement.
\end{proof}

\section{Section B: Existence and uniqueness of bridge functions}
\label{section b: bridge}
\subsection*{Existence}
Similar to \cite{Ying2023ProximalStudies} and \cite{Miao2018IdentifyingConfounder}, we consider the
singular value decomposition (\cite{Carrasco2007}, Theorem 2.41) of compact operators to characterize conditions for existence of a solutions to the integral equations from Equation \ref{eq: obs_bridge_functions}. Let $L_2\{F(t)\}$ denote the space of all square integrable functions of t with respect to a cumulative distribution function $F(t)$, which is a Hilbert space with inner product $\langle g, h \rangle = \int g(t) h(t) dF(t)$. Define $T_{\overline{a}_k,\overline{y}_{k-1}}$ as the conditional expectation operator, mapping $L_2\{F(\overline{w}_{k} \mid \overline{a}_k,\overline{y}_{k-1})\} \to L_2\{F(\overline{z}_{k} \mid \overline{a}_k,\overline{y}_{k-1})\}$ and $T_{\overline{a}_k,\overline{y}_{k-1}} h = E[h(\overline{W}_{k}) \mid \overline{z}_k, \overline{a}_k,\overline{y}_{k-1}]$. Let $(\lambda_{\overline{a}_k,\overline{y}_{k-1},l}, \psi_{\overline{a}_k,\overline{y}_{k-1},l}, \phi_{\overline{a}_k,\overline{y}_{k-1},l})_{l=1}^\infty$ denote a singular value decomposition for $T_{\overline{a}_k,\overline{y}_{k-1}}$. We assume the following:

\begin{assumption}[Regularity conditions for existence]
\label{assumption: existence}
\  \\
(i) $\sum_{k=1}^2 \int \int f(\overline{w}_{k} \mid \overline{z}_k, \overline{a}_k,\overline{y}_{k-1})f(\overline{z}_{k} \mid \overline{w}_k, \overline{a}_k,\overline{y}_{k-1}) d\overline{w}_k d\overline{z}_k < \infty$. \\
(ii) For fixed $y_0,y_1,y_2,a_1,a_2$, 
\begin{equation*}
\begin{aligned}
&\int f^2(y_2 \mid \overline{z}_2, \overline{a}_2, \overline{y}_1) f(\overline{z}_2 \mid \overline{a}_2, \overline{y}_1) d\overline{z}_2 < \infty, \\ & \int \left(\sum_{\overline{w}_{2}}h_{22}\left( \overline{y}_{2},\overline{w}_{2},%
\overline{a}_{2}\right) f\left( \overline{w}_{2}|\overline{y}%
_{1},a_{1},z_{1}\right) f(y_{1}|y_{0},
a_{1},z_{1})\right)^2 f(z_1 \mid a_1,y_0) dz_1 < \infty.
\end{aligned}
\end{equation*} \\
(iii) For fixed $y_0,y_1,y_2,a_1,a_2$, 
\begin{equation*}
\begin{aligned}
& \sum_{l=1}^\infty \lambda^{-2}_{\overline{a}_2,\overline{y}_{1},l} \langle f(y_2 \mid \overline{z}_2, \overline{a}_2, \overline{y}_1), \phi_{\overline{a}_2,\overline{y}_{1},l} \rangle^2 < \infty, \\ & \sum_{l=1}^\infty \lambda^{-2}_{a_1,y_0,l} \left \langle \sum_{\overline{w}_{2}}h_{22}\left( \overline{y}_{2},\overline{w}_{2},%
\overline{a}_{2}\right) f\left( \overline{w}_{2}|\overline{y}%
_{1},a_{1},z_{1}\right) f(y_{1}|y_{0},
a_{1},z_{1}), \phi_{a_1,y_0,l} \right \rangle^2 < \infty.
\end{aligned}
\end{equation*}
\\
(iv) For any square integrable $g_1$, we have that 
\begin{equation*}
    {E}[g_1(\overline{Z}_2) \mid \overline{y}_1, \overline{w}_2, \overline{a}_2] = 0 \text{ implies } g_1(\overline{Z}_2) = 0.
\end{equation*}
In addition, for any square integrable $g_0$, we have that 
\begin{equation*}
    {E}[g_0(Z_1) \mid {y}_0, w_1, {a}_1] = 0 \text{ implies } g_0(Z_1) =0. 
\end{equation*}
\end{assumption}
\begin{lemma}
Under Assumption \ref{assumption: existence}, there exist $h_{22}$ and $h_{21}$ that satisfy the respective integral equations from Equation \ref{eq: obs_bridge_functions} from Lemma \ref{lm: link_latent_to_observed}.
\end{lemma}
\begin{proof}
The proof follows straightforwardly from Picard's theorem (\cite{Kress1999}) and Lemma 2 from \cite{Miao2018IdentifyingConfounder}.
\end{proof}

\subsection*{Uniqueness}
In this subsection, we introduce an additional completeness assumption that guarantees that solutions to the observed integral equations are unique.
\begin{assumption}
\label{assumption: completeness for uniqueness}
For any square integrable $g_1$, we have that 
\begin{equation*}
    {E}(g_1(\overline{W}_2) \mid \overline{y}_1, \overline{z}_2, \overline{a}_2) = 0 \text{ implies } g_1(\overline{W}_2) = 0.
\end{equation*}
In addition, for any square integrable $g_0$, we have that 
\begin{equation*}
    {E}(g_0(W_1) \mid {y}_0, z_1, {a}_1) = 0 \text{ implies } g_0(W_1) =0. 
\end{equation*}
\end{assumption}
\begin{lemma}
\label{lm: uniqueness}
Suppose there exist $h_{22}$ and $h_{21}$ that solve the intergal equations from Lemma \ref{lm: link_latent_to_observed}. Then under Assumption \ref{assumption: completeness for uniqueness}, $h_{22}$ and $h_{21}$ are unique.
\begin{proof}
First, suppose that $h_{22}$ and $h'_{22}$ both solve the integral equation. Then 
\begin{equation*}
\sum_{\overline{w}_{2}}h_{22}\left( \overline{y}_{2},\overline{w%
}_{2},\overline{a}_{2}\right) f\left( \overline{w}_{2}|\overline{y}_{1},%
\overline{a}_{2},\overline{z}_{2}\right)  = \sum_{\overline{w}_{2}}h'_{22}\left( \overline{y}_{2},\overline{w%
}_{2},\overline{a}_{2}\right) f\left( \overline{w}_{2}|\overline{y}_{1},%
\overline{a}_{2},\overline{z}_{2}\right).
\end{equation*}
Subtracting, we get 
\begin{equation*}
{E}[h_{22}\left( \overline{y}_{2},\overline{W
}_{2},\overline{a}_{2}\right) - h'_{22}\left( \overline{y}_{2},\overline{W
}_{2},\overline{a}_{2}\right) |\overline{y}_{1},%
\overline{a}_{2},\overline{z}_{2}] = 0 \implies h_{22}\left( \overline{y}_{2},\overline{W
}_{2},\overline{a}_{2}\right) - h'_{22}\left( \overline{y}_{2},\overline{W
}_{2},\overline{a}_{2}\right) = 0.
\end{equation*}
Next, suppose that $h_{21}$ and $h'_{21}$ both solve the corresponding integral equation. Then by the uniqueness of $h_{22}$, we get 
\begin{equation*}
\sum_{w_{1}}h_{21}\left( \overline{y}%
_{2},w_{1},\overline{a}_{2}\right) f\left( w_{1}|y_{0},a_{1},z_{1}\right) = \sum_{w_{1}}h'_{21}\left( \overline{y}%
_{2},w_{1},\overline{a}_{2}\right) f\left( w_{1}|y_{0},a_{1},z_{1}\right).
\end{equation*}
Subtracting, we get 
\begin{equation*}
{E}[h_{21}\left( \overline{y}
_{2},W_{1},\overline{a}_{2}\right)-h'_{21}\left( \overline{y}%
_{2},W_{1},\overline{a}_{2}\right)|y_{0},a_{1},z_{1}] = 0 \implies h_{21}\left( \overline{y}
_{2},W_{1},\overline{a}_{2}\right)-h'_{21}\left( \overline{y}%
_{2},W_{1},\overline{a}_{2}\right) = 0.
\end{equation*}
\end{proof}
\end{lemma}

\section{Section C: General number of timepoints}
\label{section c: timepoints}
In this section, we demonstrate that the results from previous sections can be extended the case where the number of time periods $K > 2$. The data has the following form: $$(U_0,\ldots,U_{K-1},Y_0,\ldots,Y_K,A_1,\ldots,A_K,Z_1,\ldots,Z_K,W_1,\ldots,W_K),$$
where the $U$ variables are unobserved. We first introduce the general versions of the latent randomization, bridge function, and completeness assumptions.
\begin{assumption}[Sequential Proximal Latent Randomization - General Case]
\label{assumption: proxy-general}
\begin{eqnarray*}
&&\left\{ {Y}_{1}\left( a_{1}\right) ,\ldots,Y_{k}\left(
a_{1},\ldots a_{k}\right) ,\overline{W}_{k}\right\} \indep A_{k},\overline{Z}_{k}|\left(
\overline{U}_{k-1},\overline{Y}_{k-1},\overline{A}_{k-1} = \overline{a}_{k-1}\right) \text{ for } k = 1,\ldots,K. \\
&& \{\overline{W}_{k+1},Y_k\} \indep \overline{Z}_{k} \mid \overline{A}_{k}, \overline{U}_{k-1},\overline{Y}_{k-1} \text{ for } k = 1,\ldots,K-1.
\end{eqnarray*}  
    
\end{assumption}

\begin{assumption}[Bridge Functions - General Case]
\label{assumption: bridge_functions_general}
There exist functions $h_{KK}(\Bar{y}_K,\Bar{w}_K,\Bar{a}_K), \ldots, h_{K1}(\Bar{y}_K,\Bar{w}_1,\Bar{a}_K)$ that satisfy
\begin{equation*}
    f(Y_K = y_K \mid \Bar{y}_{K-1},\Bar{u}_{K-1},\Bar{a}_{K}) = \sum_{\Bar{w}_K}h_{KK}(\Bar{y}_K,\Bar{w}_K,\Bar{a}_K)f(\Bar{w}_K \mid \Bar{y}_{K-1}, \Bar{u}_{K-1}, \Bar{a}_K),
\end{equation*}
and for $k = 1, \ldots, K-1$,
\begin{equation*}
    \sum_{\Bar{w}_{k+1}} h_{K,k+1}(\Bar{y}_K,\Bar{w}_{k+1},\Bar{a}_K)f(\Bar{w}_{k+1},y_{k} \mid \Bar{y}_{k-1}, \Bar{u}_{k-1},\Bar{a}_k)= \sum_{\Bar{w}_{k}}h_{Kk}(\Bar{y}_K,\Bar{w}_{k},\Bar{a}_K) f(\Bar{w}_{k} \mid \Bar{y}_{k-1}, \Bar{u}_{k-1},\Bar{a}_k).
\end{equation*}
\end{assumption}
\begin{assumption}[Completeness - General Case]
\label{assumption: completeness-general}
For each $k = 1,\ldots,K$ and any square integrable $g = g(\overline{U}_{k-1})$,
\begin{equation*}
E\{g(\overline{U}_{k-1} \mid \overline{y}_{k-1},\overline{z}_{k},\overline{a}_{k}\} \text{ implies } g(\overline{U}_{k-1}) = 0.
\end{equation*}
\end{assumption}
\begin{lemma}
\label{lm: g_one_step}
Under Assumptions \ref{assumption: proxy-general} and \ref{assumption: bridge_functions_general}, the following equation holds:
\begin{equation*}
\begin{aligned}
\sum_{\Bar{w}_{k}} &h_{Kk}(\Bar{y}_K,\Bar{w}_{k},\Bar{a}_K) f(\Bar{w}_{k} \mid \Bar{y}_{k-1}, \Bar{u}_{k-1},\Bar{a}_k) \\ &= \sum_{u_k} \sum_{\Bar{w}_{k+1}}h_{K,k+1}(\Bar{y}_K,\Bar{w}_{k+1},\Bar{a}_K) f(\Bar{w}_{k+1} \mid \Bar{y}_{k}, \Bar{u}_{k},\Bar{a}_{k+1}) f(u_k \mid \Bar{u}_{k-1},\Bar{y}_k, \Bar{a}_k) f(y_k \mid \Bar{u}_{k-1},\Bar{y}_{k-1}, \Bar{a}_k).
\end{aligned}
\end{equation*}
\end{lemma}
\begin{proof}
\begin{equation*}
\begin{aligned}
\sum_{\Bar{w}_{k}} &h_{Kk}(\Bar{y}_K,\Bar{w}_{k},\Bar{a}_K) f(\Bar{w}_{k} \mid \Bar{y}_{k-1}, \Bar{u}_{k-1},\Bar{a}_k) \\ &= \sum_{\Bar{w}_{k+1}} h_{K,k+1}(\Bar{y}_K,\Bar{w}_{k+1},\Bar{a}_K)f(\Bar{w}_{k+1},y_{k} \mid \Bar{y}_{k-1}, \Bar{u}_{k-1},\Bar{a}_k) \\ &= \sum_{\Bar{w}_{k+1}} h_{K,k+1}(\Bar{y}_K,\Bar{w}_{k+1},\Bar{a}_K)f(\Bar{w}_{k+1} \mid \Bar{y}_{k}, \Bar{u}_{k-1},\Bar{a}_k) f(y_k \mid \Bar{y}_{k-1}, \Bar{u}_{k-1},\Bar{a}_k) \\ &= \sum_{\Bar{w}_{k+1}} \sum_{u_k} h_{K,k+1}(\Bar{y}_K,\Bar{w}_{k+1},\Bar{a}_K)f(\Bar{w}_{k+1}, u_k \mid \Bar{y}_{k}, \Bar{u}_{k-1},\Bar{a}_k) f(y_k \mid \Bar{y}_{k-1}, \Bar{u}_{k-1},\Bar{a}_k) \\ &= \sum_{\Bar{w}_{k+1}} \sum_{u_k} h_{K,k+1}(\Bar{y}_K,\Bar{w}_{k+1},\Bar{a}_K)f(\Bar{w}_{k+1} \mid \Bar{y}_{k}, \Bar{u}_{k},\Bar{a}_k) f(u_k \mid \Bar{y}_{k}, \Bar{u}_{k-1},\Bar{a}_k)  f(y_k \mid \Bar{y}_{k-1}, \Bar{u}_{k-1},\Bar{a}_k) \\ &= \sum_{u_k} \sum_{\Bar{w}_{k+1}}h_{K,k+1}(\Bar{y}_K,\Bar{w}_{k+1},\Bar{a}_K) f(\Bar{w}_{k+1} \mid \Bar{y}_{k}, \Bar{u}_{k},\Bar{a}_{k+1}) f(u_k \mid \Bar{u}_{k-1},\Bar{y}_k, \Bar{a}_k) f(y_k \mid \Bar{u}_{k-1},\Bar{y}_{k-1}, \Bar{a}_k).
\end{aligned}
\end{equation*}
By assumption \ref{assumption: bridge_functions_general}, the first equality holds. The second, third, and fourth equality follow from standard probability algebra, and the last follows from $A_{k+1} \indep \Bar{W}_{k+1} \mid \Bar{Y}_{k}, \Bar{U}_{k}, \overline{A}_{k} = \overline{a}_{k}$.
\end{proof}
\begin{theorem}
Under Assumptions \ref{assumption: proxy-general} and \ref{assumption: bridge_functions_general},
\begin{equation*}
f(Y_K(a_1,\ldots,a_K) = y_K, \ldots, Y_1(a_1) = y_1 \mid y_0) = \sum_{w_1}h_{K1}(\Bar{y}_K, w_1,\Bar{a}_K) f(w_1 \mid y_0).
\end{equation*}
\end{theorem}
\begin{proof}
By Lemma \ref{lm: g_one_step},
\begin{equation*}
\begin{aligned}
\sum_{w_1}&h_{K1}(\Bar{y}_K, w_1,\Bar{a}_K) f(w_1 \mid y_0) = \sum_{u_0}\sum_{w_1}h_{K1}(\Bar{y}_K, w_1,\Bar{a}_K) f(w_1 \mid y_0, u_0, a_1) f(u_0 \mid y_0) \\ &= \sum_{u_0} \sum_{u_1} \sum_{\Bar{w}_2}h_{K2}(\Bar{y}_K, \Bar{w}_2,\Bar{a}_K) f(\Bar{w}_2 \mid \Bar{y}_1, \Bar{u}_1, \Bar{a}_2) f(u_1 \mid u_0, \Bar{y}_1, a_1) f(y_1 \mid u_0, y_0, a_1) f(u_0 \mid y_0) \\  &\vdots\\ &= \sum_{\Bar{u}_{K-1}} \sum_{\Bar{w}_K} h_{KK}(\Bar{y}_K, \Bar{w}_K,\Bar{a}_K) f(\Bar{w}_K \mid \Bar{y}_{K-1}, \Bar{u}_{K-1}, \Bar{a}_K) \left\{\prod_{j=1}^{K-1} f(u_j \mid \Bar{u}_{j-1},\Bar{y}_j,\Bar{a}_j)f(y_j \mid \Bar{u}_{j-1},\Bar{y}_{-1},\Bar{a}_j) \right\} f(u_0 \mid y_0) \\ &= \sum_{\Bar{u}_{K-1}} f(Y_K = y_K \mid \Bar{y}_{K-1}, \Bar{u}_{K-1}, \Bar{a}_K) \left\{\prod_{j=1}^{K-1} f(u_j \mid \Bar{u}_{j-1},\Bar{y}_j,\Bar{a}_j)f(y_j \mid \Bar{u}_{j-1},\Bar{y}_{-1},\Bar{a}_j) \right\} f(u_0 \mid y_0) \\ &= f(Y_K(a_1,\ldots,a_K) = y_K, \ldots, Y_1(a_1) = y_1 \mid y_0).
\end{aligned}
\end{equation*}
The first equality follows by basic probability, the inner equalities all follow from Lemma \ref{lm: g_one_step}, the second to last equality follows from Assumption \ref{assumption: bridge_functions_general}, and the last follows by the usual g-formula, since we have sequential exchangeability given the $Y$ and $U$ variables from Assumption \ref{assumption: proxy-general}.
\end{proof}
In the following lemma, similar to the $K = 2$ case, we link the latent bridge functions to observable bridge functions.
\begin{lemma}
 Suppose that there exist functions $h_{KK},\ldots,h_{K1}$ that satisfy 
 \begin{equation*}
    f(Y_K = y_K \mid \Bar{y}_{K-1},\Bar{z}_{K},\Bar{a}_{K}) = \sum_{\Bar{w}_K}h_{KK}(\Bar{y}_K,\Bar{w}_K,\Bar{a}_K)f(\Bar{w}_K \mid \Bar{y}_{K-1}, \Bar{z}_{K}, \Bar{a}_K),
\end{equation*}
and for $k = 1, \ldots, K-1$,
\begin{equation*}
    \sum_{\Bar{w}_{k+1}} h_{K,k+1}(\Bar{y}_K,\Bar{w}_{k+1},\Bar{a}_K)f(\Bar{w}_{k+1},y_{k} \mid \Bar{y}_{k-1}, \Bar{z}_{k},\Bar{a}_k)= \sum_{\Bar{w}_{k}}h_{Kk}(\Bar{y}_K,\Bar{w}_{k},\Bar{a}_K) f(\Bar{w}_{k} \mid \Bar{y}_{k-1}, \Bar{z}_{k},\Bar{a}_k).
\end{equation*}
If Assumption \ref{assumption: proxy-general} and \ref{assumption: completeness-general} hold, then $h_{KK},\ldots,h_{K1}$ solve the respective integral equations from Assumption \ref{assumption: bridge_functions_general} as well.
\end{lemma}
\begin{proof}
Fix a $k$. Suppose that
\begin{equation*}
    \sum_{\Bar{w}_{k+1}} h_{K,k+1}(\Bar{y}_K,\Bar{w}_{k+1},\Bar{a}_K)f(\Bar{w}_{k+1},y_{k} \mid \Bar{y}_{k-1}, \Bar{z}_{k},\Bar{a}_k)= \sum_{\Bar{w}_{k}}h_{Kk}(\Bar{y}_K,\Bar{w}_{k},\Bar{a}_K) f(\Bar{w}_{k} \mid \Bar{y}_{k-1}, \Bar{z}_{k},\Bar{a}_k).
\end{equation*}
We can simplify using the fact that $\{\overline{W}_{k+1}, Y_k\} \indep \overline{Z}_k \mid \overline{Y}_{k-1}, \overline{A}_k, \overline{U}_{k-1}$ and $\overline{W}_{k} \indep \overline{Z}_k \mid \overline{Y}_{k-1}, \overline{A}_k, \overline{U}_{k-1}$ from Assumption \ref{assumption: proxy-general}:
\begin{equation*}
\begin{aligned}
f(\Bar{w}_{k+1},y_{k} \mid \Bar{y}_{k-1}, \Bar{z}_{k},\Bar{a}_k) &= \sum_{\overline{u}_{k-1}} f(\Bar{w}_{k+1},y_{k} \mid \Bar{y}_{k-1}, \Bar{z}_{k},\Bar{a}_k, \overline{u}_{k-1}) f(\overline{u}_{k-1} \mid \Bar{y}_{k-1}, \Bar{z}_{k},\Bar{a}_k) \\ &= \sum_{\overline{u}_{k-1}} f(\Bar{w}_{k+1},y_{k} \mid \Bar{y}_{k-1}, \Bar{a}_k, \overline{u}_{k-1}) f(\overline{u}_{k-1} \mid \Bar{y}_{k-1}, \Bar{z}_{k},\Bar{a}_k), \\ 
f(\Bar{w}_{k} \mid \Bar{y}_{k-1}, \Bar{z}_{k},\Bar{a}_k) &= \sum_{\overline{u}_{k-1}} f(\Bar{w}_{k} \mid \Bar{y}_{k-1}, \Bar{z}_{k},\Bar{a}_k, \overline{u}_{k-1}) f(\overline{u}_{k-1} \mid \Bar{y}_{k-1}, \Bar{z}_{k},\Bar{a}_k) \\ &= \sum_{\overline{u}_{k-1}} f(\Bar{w}_{k} \mid \Bar{y}_{k-1}, \Bar{z}_{k}, \overline{u}_{k-1}) f(\overline{u}_{k-1} \mid \Bar{y}_{k-1}, \Bar{z}_{k},\Bar{a}_k).
\end{aligned}
\end{equation*}
Plugging these equations back into the previous display and changing the order of summation, we get 
\begin{equation*}
\begin{aligned}
    \sum_{\overline{u}_{k-1}}  \sum_{\Bar{w}_{k+1}} &h_{K,k+1}(\Bar{y}_K,\Bar{w}_{k+1},\Bar{a}_K)f(\Bar{w}_{k+1},y_{k} \mid \Bar{y}_{k-1}, \Bar{a}_k, \overline{u}_{k-1}) f(\overline{u}_{k-1} \mid \Bar{y}_{k-1}, \Bar{z}_{k},\Bar{a}_k) \\&= \sum_{\overline{u}_{k-1}} \sum_{\Bar{w}_{k}}h_{Kk}(\Bar{y}_K,\Bar{w}_{k},\Bar{a}_K) f(\Bar{w}_{k} \mid \Bar{y}_{k-1}, \Bar{z}_{k}, \overline{u}_{k-1}) f(\overline{u}_{k-1} \mid \Bar{y}_{k-1}, \Bar{z}_{k},\Bar{a}_k).
\end{aligned}
\end{equation*}
By the completeness condition from Assumption \ref{assumption: completeness-general}, we get that $$\sum_{\Bar{w}_{k+1}} h_{K,k+1}(\Bar{y}_K,\Bar{w}_{k+1},\Bar{a}_K)f(\Bar{w}_{k+1},y_{k} \mid \Bar{y}_{k-1}, \Bar{a}_k, \overline{u}_{k-1}) = \sum_{\Bar{w}_{k}}h_{Kk}(\Bar{y}_K,\Bar{w}_{k},\Bar{a}_K) f(\Bar{w}_{k} \mid \Bar{y}_{k-1}, \Bar{z}_{k}, \overline{u}_{k-1}),$$
as desired. Next, suppose that
\begin{equation*}
     f(Y_K = y_K \mid \Bar{y}_{K-1},\Bar{z}_{K},\Bar{a}_{K}) = \sum_{\Bar{w}_K}h_{KK}(\Bar{y}_K,\Bar{w}_K,\Bar{a}_K)f(\Bar{w}_K \mid \Bar{y}_{K-1}, \Bar{z}_{K}, \Bar{a}_K).
\end{equation*}
Using the fact that $Y_K \indep \overline{Z}_K \mid \Bar{Y}_{K-1},\Bar{A}_{K},\overline{U}_{K-1}$ and $\overline{W}_K \indep \overline{Z}_K \mid \Bar{Y}_{K-1},\Bar{A}_{K},\overline{U}_{K-1}$ from Assumption \ref{assumption: proxy-general},
\begin{equation*}
\begin{aligned}
 f(Y_K = y_K \mid \Bar{y}_{K-1},\Bar{z}_{K},\Bar{a}_{K}) &= \sum_{\overline{u}_{K-1}} f(Y_K = y_K \mid \Bar{y}_{K-1},\Bar{z}_{K},\Bar{a}_{K},\overline{u}_{K-1})f(\overline{u}_{K-1} \mid \Bar{y}_{K-1},\Bar{z}_{K},\Bar{a}_{K})  \\ &= \sum_{\overline{u}_{K-1}} f(Y_K = y_K \mid \Bar{y}_{K-1},\Bar{a}_{K},\overline{u}_{K-1})f(\overline{u}_{K-1} \mid \Bar{y}_{K-1},\Bar{z}_{K},\Bar{a}_{K}), \\ f(\Bar{w}_K \mid \Bar{y}_{K-1}, \Bar{z}_{K}, \Bar{a}_K) &= \sum_{\overline{u}_{K-1}}f(\Bar{w}_K \mid \Bar{y}_{K-1}, \Bar{z}_{K}, \Bar{a}_K, \overline{u}_{K-1}) f(\overline{u}_{K-1} \mid \Bar{y}_{K-1},\Bar{z}_{K},\Bar{a}_{K}) \\ &= \sum_{\overline{u}_{K-1}}f(\Bar{w}_K \mid \Bar{y}_{K-1}, \Bar{a}_K, \overline{u}_{K-1}) f(\overline{u}_{K-1} \mid \Bar{y}_{K-1},\Bar{z}_{K},\Bar{a}_{K}).
\end{aligned}
\end{equation*}
Plugging these equations back into the previous display and changing the order of summation, we get
\begin{equation*}
\begin{aligned}
\sum_{\overline{u}_{K-1}} &f(Y_K = y_K \mid \Bar{y}_{K-1},\Bar{a}_{K},\overline{u}_{K-1})f(\overline{u}_{K-1} \mid \Bar{y}_{K-1},\Bar{z}_{K},\Bar{a}_{K}) \\ &= \sum_{\overline{u}_{K-1}}\sum_{\Bar{w}_K}h_{KK}(\Bar{y}_K,\Bar{w}_K,\Bar{a}_K)f(\Bar{w}_K \mid \Bar{y}_{K-1}, \Bar{a}_K, \overline{u}_{K-1}) f(\overline{u}_{K-1} \mid \Bar{y}_{K-1},\Bar{z}_{K},\Bar{a}_{K}).
\end{aligned}
\end{equation*}
Finally, the completeness condition from Assumption \ref{assumption: completeness-general} implies that $f(Y_K = y_K \mid \Bar{y}_{K-1},\Bar{a}_{K},\overline{u}_{K-1}) = \sum_{\Bar{w}_K}h_{KK}(\Bar{y}_K,\Bar{w}_K,\Bar{a}_K)f(\Bar{w}_K \mid \Bar{y}_{K-1}, \Bar{a}_K, \overline{u}_{K-1})$, completing the proof.
\end{proof}
\section{Section D: Over-identified Case}
\label{section d: overidentified}
In this section, we examine the setting where $\min(|W_t|,|Z_t|) \geq |U_{t-1}|$ for $t = 1,2$. We require a rank condition on certain matrices.
\begin{assumption}
\label{assumption: categorical_rank}
For all $(\overline{y}_1,\overline{a}_2)$, the matrices
$P(\overline{W}_2 \mid \overline{U}_1, \overline{y}_1,\overline{a}_2)$ and $P(\overline{U}_1 \mid \overline{Z}_2,\overline{y}_1,\overline{a}_2)$ have rank $|\overline{U}_1|$. In addition, for all $({y}_0,{a}_1)$, $P({W}_1 \mid {U}_0, {y}_0,{a}_1)$ and $P({U}_1 \mid {Z}_1,{y}_0,{a}_1)$ have rank $|U_0|$.
\end{assumption}
\begin{proposition}
Let $\infty > \min(|W_t|,|Z_t|) \geq |U_{t-1}|$ for $t = 1,2$. Also, suppose Assumptions \ref{assumption: proxy} and \ref{assumption: categorical_rank} hold. Then there exist solutions $h_{22}$, $h_{21}$, and $h_{11}$ solving the respective integral equations from Assumption \ref{assumption: bridge_functs_2} and Equation \ref{eq: bridge_funct_1}. In the special case where $|W_t| = |Z_t| = |U_{t-1}|$ for $t = 1,2$, the conclusion of Proposition 2 holds. 
\end{proposition}
\begin{proof}
Our proof resembles the proof of Lemma 1 from \cite{Shi2020MultiplyConfounding}. Starting from the last time step, we know that $\overline{W}_2 \indep A_2, \overline{Z}_2 \mid \overline{U}_1, \overline{Y}_1,A_1=a_1$ and $\overline{Y}_2(a_1,a_2) \indep A_2, \overline{Z}_2 \mid \overline{U}_1, \overline{Y}_1,A_1=a_1$, so 
\begin{equation} 
\label{eq: categorical h_22 ci}
\begin{aligned}
    P(\overline{W}_2 \mid \overline{Z}_2, \overline{y}_1,\overline{a}_2) &= P(\overline{W}_2 \mid \overline{U}_1, \overline{y}_1,{a}_1)P(\overline{U}_1 \mid \overline{Z}_2,\overline{y}_1,\overline{a}_2) \\
    P(y_2 \mid \overline{Z}_2,\overline{y}_1,\overline{a}_2) &= P(y_2 \mid \overline{U}_1,\overline{y}_1,\overline{a}_2)P(\overline{U}_1 \mid \overline{Z}_2,\overline{y}_1,\overline{a}_2)
\end{aligned}
\end{equation}
Recall that by Assumption \ref{assumption: categorical_rank}, $P(\overline{W}_2 \mid \overline{U}_1, \overline{y}_1,{a}_1)$ and $P(\overline{U}_1 \mid \overline{Z}_2,\overline{y}_1,\overline{a}_2)$ have rank $|\overline{U}_1|$, so they have left and right inverses, respectively. Moreover, the left and right inverses equal the Moore-Penrose pseudoinverses. We denote the Moore-Penrose pseudoinverse of a matrix $A$ by $A^\dagger$. We obtain
\begin{equation*}
\begin{aligned}
P(\overline{U}_1 \mid \overline{Z}_2,\overline{y}_1,\overline{a}_2) &= P(\overline{W}_2 \mid \overline{U}_1, \overline{y}_1,{a}_1)^{\dagger}  P(\overline{W}_2 \mid \overline{Z}_2, \overline{y}_1,\overline{a}_2), \\  P(y_2 \mid \overline{Z}_2,\overline{y}_1,\overline{a}_2) &= P(y_2 \mid \overline{U}_1,\overline{y}_1,\overline{a}_2)P(\overline{W}_2 \mid \overline{U}_1, \overline{y}_1,{a}_1)^{\dagger}P(\overline{W}_2 \mid \overline{Z}_2, \overline{y}_1,\overline{a}_2)
\end{aligned}
\end{equation*}
Thus, there exists an $h_{22}$ (which may not be unique) that solves 
\begin{equation}
\label{eq: categorical h_22}
    P(y_2 \mid \overline{Z}_2,\overline{y}_1,\overline{a}_2) = h_{22}P(\overline{W}_2 \mid \overline{Z}_2, \overline{y}_1,\overline{a}_2),
\end{equation}
where we regard the vector $h_{22}$ (of dimension $1 \times |\overline{W}_2|$), as holding the values of the function $h_{22}(\overline{y}_2,\overline{w}_2,\overline{a}_2)$ for the different values $\overline{w}_2$ can take. Since $P(\overline{U}_1 \mid \overline{Z}_2,\overline{y}_1,\overline{a}_2)$ has a right inverse, we obtain from Equation \ref{eq: categorical h_22 ci} that
\begin{equation*}
\begin{aligned}
P(\overline{W}_2 \mid \overline{Z}_2, \overline{y}_1,\overline{a}_2)P(\overline{U}_1 \mid \overline{Z}_2,\overline{y}_1,\overline{a}_2)^\dagger &= P(\overline{W}_2 \mid \overline{U}_1, \overline{y}_1,{a}_1), \\ 
P(y_2 \mid \overline{Z}_2,\overline{y}_1,\overline{a}_2)P(\overline{U}_1 \mid \overline{Z}_2,\overline{y}_1,\overline{a}_2)^\dagger &= P(y_2 \mid \overline{U}_1,\overline{y}_1,\overline{a}_2).
\end{aligned}
\end{equation*}
Thus, 
\begin{equation*}
\begin{aligned}
P(y_2 \mid \overline{U}_1,\overline{y}_1,\overline{a}_2) &= P(y_2 \mid \overline{Z}_2,\overline{y}_1,\overline{a}_2)P(\overline{U}_1 \mid \overline{Z}_2,\overline{y}_1,\overline{a}_2)^\dagger \\ &= h_{22}P(\overline{W}_2 \mid \overline{Z}_2, \overline{y}_1,\overline{a}_2)P(\overline{U}_1 \mid \overline{Z}_2,\overline{y}_1,\overline{a}_2)^\dagger \\ &= h_{22} P(\overline{W}_2 \mid \overline{U}_1, \overline{y}_1,{a}_1).
\end{aligned}
\end{equation*}

We can do something similar for the $h_{21}$ function. For a fixed $y_1$, we may also view $P(\overline{W}_2,y_1 \mid Z_1,y_0,a_1)$ as a $|W_1| \times |W_2| \text{ by }|Z_1|$ matrix. Using the conditional independence assumptions in \ref{assumption: proxy}, we know that 
\begin{equation}
\label{eq: categorical h_21 ci}
\begin{aligned}
P(W_1 \mid Z_1, y_0, a_1) &= P(W_1 \mid U_0, y_0, a_1)P(U_0 \mid Z_1, y_0, a_1) \\
P(\overline{W}_2,y_1 \mid Z_1, y_0, a_1) &= P(\overline{W}_2,y_1 \mid U_0, y_0, a_1)P(U_0 \mid Z_1, y_0, a_1)
\end{aligned}
\end{equation}
Since $ P(W_1 \mid U_0, y_0, a_1)$ has rank $|U_0|$ and thus has a left inverse, we get 
\begin{equation*}
\begin{aligned}
P(\overline{W}_2,y_1 &\mid Z_1, y_0, a_1) = \\&P(\overline{W}_2,y_1 \mid U_0, y_0, a_1)P(W_1 \mid U_0, y_0, a_1)^{\dagger}P(W_1 \mid Z_1, y_0, a_1)
\end{aligned}
\end{equation*}
which also implies 
\begin{equation*}
\begin{aligned}
h_{22}P(\overline{W}_2,y_1 &\mid Z_1, y_0, a_1) = \\&h_{22}P(\overline{W}_2,y_1 \mid U_0, y_0, a_1)P(W_1 \mid U_0, y_0, a_1)^\dagger P(W_1 \mid Z_1, y_0, a_1),
\end{aligned}
\end{equation*}
for any $h_{22}$ vector, so there exists an $h_{21}$ (which may not be unique) vector that solves 
\begin{equation}
\label{eq: categorical h_21}
\begin{aligned}
h_{22}P(\overline{W}_2,y_1 &\mid Z_1, y_0, a_1) = h_{21}P(W_1 \mid Z_1, y_0, a_1),
\end{aligned}
\end{equation}
where we regard the vector $h_{21}$ (of dimension $1 \times |{W}_1|$), as holding the values of the function $h_{21}(\overline{y}_2,{w}_1,\overline{a}_2)$ for the different values ${w}_1$ can take. We also know that since $P(U_0 \mid Z_1, y_0, a_1)$ has a right inverse and rearranging Equation \ref{eq: categorical h_21 ci},
\begin{equation*}
\begin{aligned}
    P(W_1 \mid Z_1, y_0, a_1)P(U_0 \mid Z_1, y_0, a_1)^\dagger &= P(W_1 \mid U_0, y_0, a_1) \\
P(\overline{W}_2,y_1 \mid Z_1, y_0, a_1)P(U_0 \mid Z_1, y_0, a_1)^\dagger &= P(\overline{W}_2,y_1 \mid U_0, y_0, a_1).
\end{aligned}
\end{equation*}
Thus, we get that 
\begin{equation*}
\begin{aligned}
h_{22} P(\overline{W}_2,y_1 \mid U_0, y_0, a_1) &= h_{22} P(\overline{W}_2,y_1 \mid Z_1, y_0, a_1)P(U_0 \mid Z_1, y_0, a_1)^\dagger \\ &= h_{21}P(W_1 \mid Z_1, y_0, a_1)P(U_0 \mid Z_1, y_0, a_1)^\dagger \\ &= h_{21} P(W_1 \mid U_0, y_0, a_1).
\end{aligned} 
\end{equation*}

Lastly, for $h_{11}$, the steps follow similarly as in for $h_{22}$, as we know that ${W}_1 \indep A_1, {Z}_1 \mid {U}_0, {Y}_0$ and $Y_1(a_1) \indep {Z}_1, A_1 \mid {U}_0, {Y}_0$ so 
\begin{equation}
\label{eq: categorical h_11 ci}
\begin{aligned}
    P({W}_1 \mid {Z}_1, {y}_0,a_1) &= P({W}_1 \mid {U}_0, {y}_0, {a}_1)P({U}_0 \mid {Z}_1,{y}_0,{a}_1) \\
    P(y_1 \mid {Z}_1, {y}_0,{a}_1) &= P(y_1 \mid {U}_0, {y}_0,{a}_1)P({U}_0 \mid {Z}_1,{y}_0,{a}_1)
\end{aligned}
\end{equation}
Recall that $P({W}_1 \mid {U}_0, {y}_0,{a}_1)$ and $P({U}_1 \mid {Z}_1,{y}_0,{a}_1)$ have rank $|U_0|$ and thus have left and right inverses. Then 
\begin{equation*}
\begin{aligned}
P(U_0 \mid {Z}_1,{y}_0,{a}_1) &= P({W}_1 \mid {U}_0, {y}_0,{a}_1)^\dagger  P({W}_1 \mid {Z}_1, {y}_0,{a}_1), \\  P(y_1 \mid {Z}_1,{y}_0,{a}_1) &= P(y_1 \mid {U}_0,{y}_0,{a}_1)P({W}_1 \mid {U}_0, {y}_0,{a}_1)^\dagger P({W}_1 \mid {Z}_1, {y}_0,{a}_1),
\end{aligned}
\end{equation*}
and so there is an $h_{11}$ vector (which may not be unique) that solves
\begin{equation}
\label{eq: categorical h_11}
P(y_1 \mid {Z}_1,{y}_0,{a}_1) = h_{11} P({W}_1 \mid {Z}_1, {y}_0,{a}_1),
\end{equation}
where we regard the vector $h_{11}$ (of dimension $1 \times |{W}_1|$), as holding the values of the function $h_{11}(\overline{y}_1,{w}_1,\overline{a}_1)$ for the different values ${w}_1$ can take.
We also have that from rearranging Equation \ref{eq: categorical h_11 ci},
\begin{equation*}
\begin{aligned}
    P({W}_1 \mid {Z}_1, {y}_0,a_1)P({U}_0 \mid {Z}_1,{y}_0,{a}_1)^\dagger &= P({W}_1 \mid {U}_0, {y}_0, {a}_1) \\
    P(y_1 \mid {Z}_1, {y}_0,{a}_1)P({U}_0 \mid {Z}_1,{y}_0,{a}_1)^\dagger &= P(y_1 \mid {U}_0, {y}_0,{a}_1).
\end{aligned}
\end{equation*}
Thus,
\begin{equation*}
\begin{aligned}
P(y_1 \mid {U}_0, {y}_0,{a}_1) &= P(y_1 \mid {Z}_1, {y}_0,{a}_1)P({U}_0 \mid {Z}_1,{y}_0,{a}_1)^\dagger \\ &= h_{11} P({W}_1 \mid {Z}_1, {y}_0,{a}_1)P({U}_0 \mid {Z}_1,{y}_0,{a}_1)^\dagger \\ &= h_{11} P({W}_1 \mid {U}_0, {y}_0).
\end{aligned}
\end{equation*}
These results show that in the categorical confounding scenario, there exist $h_{22}, h_{21}, h_{11}$ that solve the equations from Assumption \ref{assumption: bridge_functs_2} and Equation \ref{eq: bridge_funct_1}. In the special case where the cardinality of the $W$, $U$, and $Z$ variables are the same, the Moore-Penrose pseudoinverses are simply inverses. Then by rearranging equations \ref{eq: categorical h_22}, \ref{eq: categorical h_21}, \ref{eq: categorical h_11}, the $h_{22}$, $h_{21}$, and $h_{11}$ vectors have unique representations
\begin{equation*}
\begin{aligned}
h_{22} &= P(y_2 \mid \overline{Z}_2,\overline{y}_1,\overline{a}_2)P(\overline{W}_2 \mid \overline{Z}_2, \overline{y}_1,\overline{a}_2)^{-1}\\
h_{21} &= h_{22}P(\overline{W}_2,y_1 \mid Z_1, y_0, a_1)P(W_1 \mid Z_1, y_0, a_1)^{-1}\\
h_{11} &= P(y_1 \mid {Z}_1,{y}_0,{a}_1)P({W}_1 \mid {Z}_1, {y}_0,{a}_1)^{-1}
\end{aligned}
\end{equation*}
Putting everything together and applying Theorem \ref{thm: counterfactual_id} we get that 
\begin{equation*}
\begin{aligned}
&P(Y_2(a_1,a_2)=y_2,Y_1(a_1)=y_1 \mid y_0) \\ &= P(y_2 \mid \overline{Z}_2,\overline{y}_1,\overline{a}_2)P(\overline{W}_2 \mid \overline{Z}_2, \overline{y}_1,\overline{a}_2)^{-1}P(\overline{W}_2,y_1 \mid Z_1, y_0, a_1) P(W_1 \mid Z_1, y_0, a_1)^{-1}P(W_1 \mid  y_0),
\end{aligned}
\end{equation*}
and 
\begin{equation*}
\begin{aligned}
&P(Y_1(a_1)=y_1 \mid y_0) \\ &= P(y_1 \mid \overline{Z}_1,\overline{y}_0,\overline{a}_1) P(W_1 \mid Z_1, y_0, a_1)^{-1}P(W_1 \mid  y_0).
\end{aligned}
\end{equation*}
This demonstrates the result of Proposition 2.
\end{proof}

\section*{Section F: Treatment bridge case}
\subsection*{Proof of Theorem \ref{thm: counterfactual_id_treat}}
\begin{proof}
By the discreteness of $Y_1$ and $Y_2$, we may rewrite $f\left( Y_{2}\left( a_{1},a_{2}\right) =y_{2},Y_{1}\left( a_{1}\right)
=y_{1}|y_{0}\right) = E[\mathbf{1}\{Y_2(a_1,a_2) = y_2,Y_1(a_1) = y_1\}\mid y_0]$. By the latent sequential randomization, we know that 
\begin{equation*}
E[\mathbf{1}\{Y_2(a_1,a_2) = y_2,Y_1(a_1) = y_1\}\mid y_0] = E \Big[ \frac{\mathbf{1}\{Y_2(a_1,a_2) = y_2,Y_1(a_1) = y_1\}\mathbf{1}\{\overline{A}_2 = \overline{a}_2\}}{P(A_2 = a_2 \mid A_1, \overline{Y}_1, \overline{U}_1)P(A_1 = a_1 \mid  Y_0, {U}_0)} \mid y_0\Big].
\end{equation*}
For a derivation of this fact, see \cite{Tsiatis2019DynamicRegimes} page 224. The proof now proceeds similarly to Theorem 2 of \cite{Ying2023ProximalStudies}.
\begin{equation*}
\begin{aligned}
 &E[q_{12}\left( \overline{y}_{1},\overline{Z}_{2},%
\overline{a}_{2}\right)\mathbf{1}\{Y_2 = y_2,Y_1 = y_1\}\mathbf{1}\{\overline{A}_2 = \overline{a}_2\} \mid y_0] \\&= E[E[q_{12}\left( \overline{y}_{1},\overline{Z}_{2},%
\overline{a}_{2}\right)\mathbf{1}\{Y_2 = y_2,Y_1 = y_1\}\mathbf{1}\{\overline{A}_2 = \overline{a}_2\}\mid \overline{A}_2, \overline{U}_1,\overline{Y}_1 ]\mid y_0] \\&= E[\mathbf{1}\{Y_2 = y_2,Y_1 = y_1\}\mathbf{1}\{\overline{A}_2 = \overline{a}_2\}E[q_{12}\left( \overline{y}_{1},\overline{Z}_{2},%
\overline{a}_{2}\right)\mid \overline{A}_2, \overline{U}_1,\overline{Y}_1 ]\mid y_0] \\&= E[\mathbf{1}\{Y_2(a_1,a_2) = y_2,Y_1 = y_1\}\mathbf{1}\{\overline{A}_2 = \overline{a}_2\}\frac{E[q_{11}(Y_0,Z_1,A_1)\mid A_1,U_0,Y_0]}{P(A_2 = a_2 \mid A_1,\overline{U}_1,\overline{Y}_1)}\mid y_0]\\&= E[E[\mathbf{1}\{Y_2(a_1,a_2) = y_2,Y_1 = y_1\}\mathbf{1}\{\overline{A}_2 = \overline{a}_2\}\frac{E[q_{11}(Y_0,Z_1,A_1)\mid A_1,U_0,Y_0]}{P(A_2 = a_2 \mid A_1,\overline{U}_1,\overline{Y}_1)} \mid A_1,\overline{U}_1,\overline{Y}_1]\mid y_0] \\&= E[\mathbf{1}\{{A}_1 = {a}_1, Y_1 = y_1\} \frac{E[q_{11}(Y_0,Z_1,A_1)\mid A_1,U_0,Y_0]}{P(A_2 = a_2 \mid A_1,\overline{U}_1,\overline{Y}_1)}E[\mathbf{1}\{Y_2(a_1,a_2) = y_2\}\mathbf{1}\{{A}_2 = {a}_2\} \mid A_1,\overline{U}_1,\overline{Y}_1]\mid y_0]  \\&= E[\mathbf{1}\{{A}_1 = {a}_1\} E[q_{11}(Y_0,Z_1,A_1)\mid A_1,U_0,Y_0]\mathbf{1}\{Y_2(a_1,a_2) = y_2, Y_1 = y_1\} \mid y_0] \\&= E[\mathbf{1}\{{A}_1 = {a}_1\} E[q_{11}(Y_0,Z_1,A_1)\mid A_1,U_0,Y_0]\mathbf{1}\{Y_2(a_1,a_2) = y_2, Y_1(a_1) = y_1\} \mid y_0] \\&= E[ \frac{1}{P(A_1 = a_1\mid U_0,Y_0)}E[\mathbf{1}\{Y_2(a_1,a_2) = y_2, Y_1(a_1) = y_1\} \mathbf{1}\{{A}_1 = {a}_1\}\mid U_0,Y_0] \mid y_0] \\&= E[ \frac{1}{P(A_1 = a_1\mid U_0,Y_0)} E[\mathbf{1}\{{A}_1 = {a}_1\}\mid U_0,Y_0]E[\mathbf{1}\{Y_2(a_1,a_2) = y_2, Y_1(a_1) = y_1\} \mid U_0,Y_0] \mid y_0] \\&= E[ \mathbf{1}\{Y_2(a_1,a_2) = y_2, Y_1(a_1) = y_1\} \mid y_0]
\end{aligned}
\end{equation*}
The first equality is by iterated expectation. The second by independence of $\overline{Z}_2$ and $Y_2$ given $\overline{A}_2, \overline{U}_1, \overline{Y}_1$. The third is by consistency and definition of $q_{12}$. The fourth is by iterated expectation. The fifth pulls out constants. The sixth is by independence between $A_2$ and $Y_2(a_1,a_2)$ conditional on $A_1, \overline{U}_1, \overline{Y}_1$ and cancellation. The seventh is by consistency. The eighth is by definition of $q_{11}$ and iterated expectation. The ninth is by independence of $A_1$ from potential outcomes $(Y_2(a_1,a_2),Y_1(a_1))$ given $U_0, Y_0$. The last is from iterated expectation. 
\end{proof}
\subsection*{Proof of Lemma \ref{lm: link_latent_to_observed_treatment}}
\begin{proof}
The proof logic follows that of Theorem B.2 from the supplementary material of \cite{Ying2023ProximalStudies} and is thus omitted. 
\end{proof}

\section{Section E: Additional details about the simulation}
\label{section e: simulation}
\subsection*{Proxy estimator}
The proxy estimator(s) for the categorical case in the simulation simply use plug-in estimates for all of the quantities according to Proposition 
2. Explicitly, 
\begin{equation*}
\begin{aligned}
&\widehat{d}_{2}^{proxy}\left( \overline{y}_{1},a_{1}\right) \\ &=\underset{a_{2}}{\arg \max } \sum_{y_{2}}\frac{\widehat{P}(y_2 \mid \overline{Z}_2,\overline{y}_1,\overline{a}_2)\widehat{P}(\overline{W}_2 \mid \overline{Z}_2, \overline{y}_1,\overline{a}_2)^{-1}\widehat{P}(\overline{W}_2,y_1 \mid Z_1, y_0, a_1) \widehat{P}(W_1 \mid Z_1, y_0, a_1)^{-1}\widehat{P}(W_1 \mid  y_0)}{\widehat{P}(y_1 \mid \overline{Z}_1,\overline{y}_0,\overline{a}_1) \widehat{P}(W_1 \mid Z_1, y_0, a_1)^{-1}\widehat{P}(W_1 \mid  y_0)}y_{2}, \\
&\widehat{V}_{2}^{proxy}\left( a_{1},\overline{y}_{1}\right) \\ &= \underset{a_{2}}{\max } \sum_{y_{2}}\frac{\widehat{P}(y_2 \mid \overline{Z}_2,\overline{y}_1,\overline{a}_2)\widehat{P}(\overline{W}_2 \mid \overline{Z}_2, \overline{y}_1,\overline{a}_2)^{-1}\widehat{P}(\overline{W}_2,y_1 \mid Z_1, y_0, a_1) \widehat{P}(W_1 \mid Z_1, y_0, a_1)^{-1}\widehat{P}(W_1 \mid  y_0)}{\widehat{P}(y_1 \mid \overline{Z}_1,\overline{y}_0,\overline{a}_1) \widehat{P}(W_1 \mid Z_1, y_0, a_1)^{-1}\widehat{P}(W_1 \mid  y_0)}y_{2}, \\
&\widehat{d}_1^{proxy}(y_0) =\underset{a_{1}}{\arg
\max } \sum_{y_{1}} P(y_1 \mid \overline{Z}_1,\overline{y}_0,\overline{a}_1) P(W_1 \mid Z_1, y_0, a_1)^{-1}P(W_1 \mid  y_0) \widehat{V}_{2}^{proxy}\left( a_{1},y_{1},y_{0}\right).
\end{aligned}
\end{equation*}
\subsection*{Oracle estimator}
We briefly review identification of the optimal regime if we had access to the unmeasured confounders. These follow directly from the g-formulas presented in the previous Section A.
The optimal regime and the optimal value function at time point 2 are as follows:
\begin{equation}
\label{eq: u_id}
\begin{aligned}
&d_{2}^{opt}\left( \overline{y}_{1},a_{1}\right)  =\underset{a_{2}}{\arg
\max }{E}\left\{ Y_{2}\left( a_{1},a_{2}\right) |\overline{Y}%
_{1}\left( a_{1}\right) =\overline{y}_{1}\right\}  \\
&=\underset{a_{2}}{\arg \max }\sum_{y_{2}}\frac{f\left( Y_{2}\left(
a_{1},a_{2}\right) =y_{2},Y_{1}\left( a_{1}\right) =y_{1}|y_{0}\right) }{%
f\left( Y_{1}\left( a_{1}\right) =y_{1}|y_{0}\right) }y_{2} \\
&=\underset{a_{2}}{\arg \max }\sum_{y_{2}}\frac{\sum_{u_{0},u_{1}}f\left(
Y_{2}=y_{2}|\overline{y}_{1},\overline{u}_{1},\overline{a}_{2}\right)
f\left(U_{1}=u_{1}|y_{1},y_{0}\ ,u_{0},a_{1}\right)
f(y_{1}|y_{0},u_{0},a_{1})f(u_{0}|y_{0})}{\sum_{u_{0}}f\left(
Y_{1}=y_{1}|a_{1},y_{0},u_{0}\right) f\left( u_{0}|y_{0}\right) }y_{2} \\
&V_{2}^{opt}\left( a_{1},\overline{y}_{1}\right)  =\underset{a_{2}}{\max }%
\sum_{y_{2}}\frac{\sum_{u_{0},u_{1}}f\left( Y_{2}=y_{2}|\overline{y}_{1},%
\overline{u}_{1},\overline{a}_{2}\right) f\left(U_{1}=u_{1}|y_{1},y_{0}\
,u_{0},a_{1}\right) f(y_{1}|y_{0},u_{0},a_{1})f(u_{0}|y_{0})}{%
\sum_{u_{0}}f\left( Y_{1}=y_{1}|a_{1},y_{0},u_{0}\right) f\left(
u_{0}|y_{0}\right) } y_{2}
\\
&d_1^{opt}(y_0)=\underset{a_{1}}{\arg
\max } {E}\left\{ V_{2}^{opt}\left( a_{1},Y_{1}\left( a_{1}\right)
,Y_{0}\right) |Y_{0}=y_{0}\right\}  \\
&=\underset{a_{1}}{\arg
\max } \sum_{y_{1}}f\left( Y_{1}\left( a_{1}\right) =y_{1}|y_{0}\right)
V_{2}^{opt}\left( a_{1},y_{1},y_{0}\right)  \\
&= \underset{a_{1}}{\arg
\max } \sum_{y_{1}}\sum_{u_{0}}f\left( Y_{1}\left( a_{1}\right)
=y_{1}|y_{0},u_{0}\right) f\left( u_{0}|y_{0}\right) V_{2}^{opt}\left(
a_{1},y_{1},y_{0}\right)  \\
&=\underset{a_{1}}{\arg
\max } \sum_{y_{1}}\sum_{u_{0}}f\left( Y_{1}=y_{1}|a_{1},y_{0},u_{0}\right)
f\left( u_{0}|y_{0}\right) V_{2}^{opt}\left( a_{1},y_{1},y_{0}\right) 
\end{aligned}
\end{equation}
The oracle in the simulation simply uses plug-in estimates for all of the quantities, i.e., 
\begin{equation}
\label{eq: oracle_estimate}
\begin{aligned}
&\widehat{d}_{2}^{oracle}\left( \overline{y}_{1},a_{1}\right) =\underset{a_{2}}{\arg \max }\sum_{y_{2}}\frac{\sum_{u_{0},u_{1}}\widehat{f}\left(
Y_{2}=y_{2}|\overline{y}_{1},\overline{u}_{1},\overline{a}_{2}\right)
\widehat{f}\left(U_{1}=u_{1}|y_{1},y_{0}\ ,u_{0},a_{1}\right)
\widehat{f}(y_{1}|y_{0},u_{0},a_{1})\widehat{f}(u_{0}|y_{0})}{\sum_{u_{0}}\widehat{f}\left(
Y_{1}=y_{1}|a_{1},y_{0},u_{0}\right) \widehat{f}\left( u_{0}|y_{0}\right) }y_{2} \\
&\widehat{V}_{2}^{oracle}\left( a_{1},\overline{y}_{1}\right)  =\underset{a_{2}}{\max }%
\sum_{y_{2}}\frac{\sum_{u_{0},u_{1}}\widehat{f}\left(
Y_{2}=y_{2}|\overline{y}_{1},\overline{u}_{1},\overline{a}_{2}\right)
\widehat{f}\left(U_{1}=u_{1}|y_{1},y_{0}\ ,u_{0},a_{1}\right)
\widehat{f}(y_{1}|y_{0},u_{0},a_{1})\widehat{f}(u_{0}|y_{0})}{\sum_{u_{0}}\widehat{f}\left(
Y_{1}=y_{1}|a_{1},y_{0},u_{0}\right) \widehat{f}\left( u_{0}|y_{0}\right) }y_{2}  
\\
&\widehat{d}_1^{oracle}(y_0) =\underset{a_{1}}{\arg
\max } \sum_{y_{1}}\sum_{u_{0}}\widehat{f}\left( Y_{1}=y_{1}|a_{1},y_{0},u_{0}\right)
\widehat{f}\left( u_{0}|y_{0}\right) \widehat{V}_{2}^{oracle}\left( a_{1},y_{1},y_{0}\right) 
\end{aligned}
\end{equation}%

\subsection*{NUCA estimator}
The NUCA estimator mistakenly assumes no unmeasured confounding, so it simply conditions on the arguments $(y_0)$ and $(y_1, a_1, y_0)$ at time points 1 and 2, respectively. The corresponding estimators can be written as follows:

\begin{equation}
\begin{aligned}
&\widehat{d}_{2}^{NUCA}\left( \overline{y}_{1},a_{1}\right) =\underset{a_{2}}{\arg \max } \ \widehat{E}(Y_2 \mid \overline{y}_{1},a_{1}) \\
&\widehat{V}_{2}^{NUCA}\left( a_{1},\overline{y}_{1}\right)  =\underset{a_{2}}{\max } \
 \widehat{E}(Y_2 \mid \overline{y}_{1},a_{1}) \\ 
 & \widehat{d}_1^{NUCA}(y_0) =\underset{a_{1}}{\arg
\max } \ \widehat{E}[\widehat{V}_{2}^{NUCA}\left( a_{1},y_{1},y_{0}\right) \mid y_0]
\end{aligned}
\end{equation}
Here, $\widehat{E}$ represents the sample average as all variables in the simulation were binary.

\subsection*{Additional simulation results}
For the methods oracle, NUCA, and proxy, we plot the histograms of regrets for each sample size in Figure \ref{fig:3-way-histogram}. We plot the difference in regret between NUCA and proxy in Figure \ref{fig:regret differences}. In addition, we display simulation results for the proxy estimators with order 4,5,7, and 8 order interactions. Specifically, over the 1000 iterations, we plot histograms of the regrets of each of the methods alongside the regrets of the NUCA estimator. The results are depicted in Figure \ref{fig: 4578}. For the most part, all seem to perform similarly, with the order 4 and 5 interaction estimators performing slightly worse than the order 7 and 8 interaction estimators.

\begin{figure}
    \centering
    \includegraphics[width=0.88\textwidth]{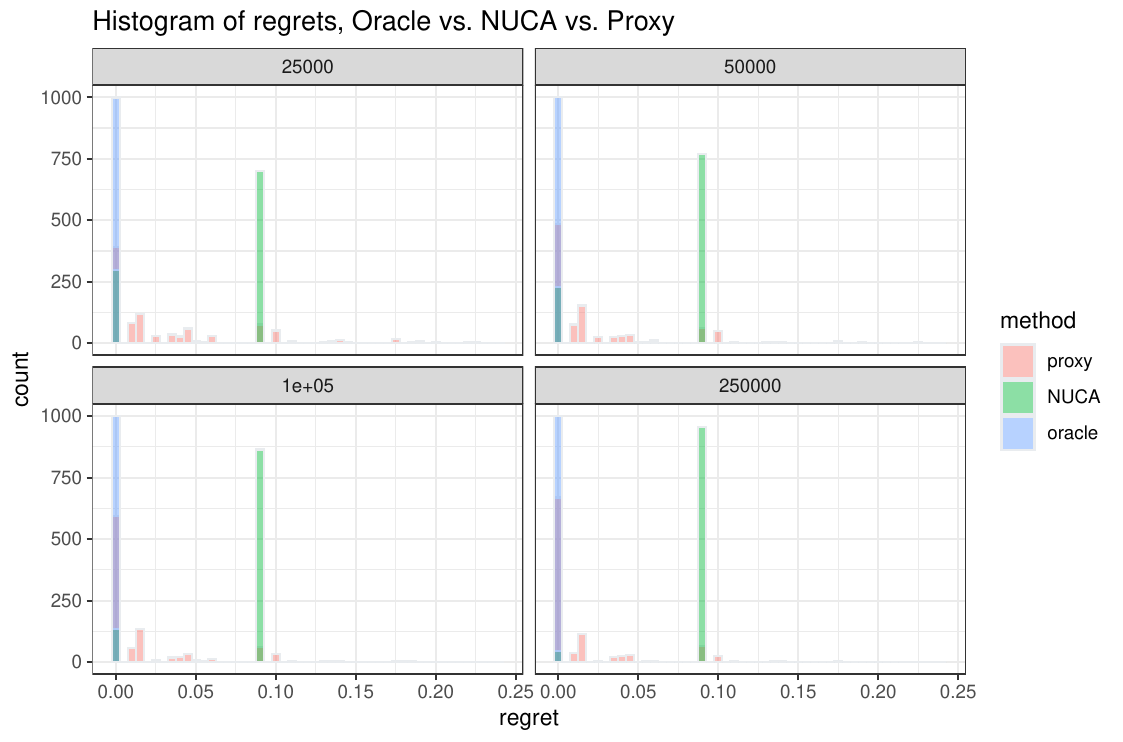}
    \caption{For each of the 4 sample sizes, we plot the regrets of the estimated optimal regime using the proxy, NUCA, and oracle methods over 1000 iterations.}
    \label{fig:3-way-histogram}
\end{figure}

\begin{figure}
    \centering
    \includegraphics[width=0.88\textwidth]{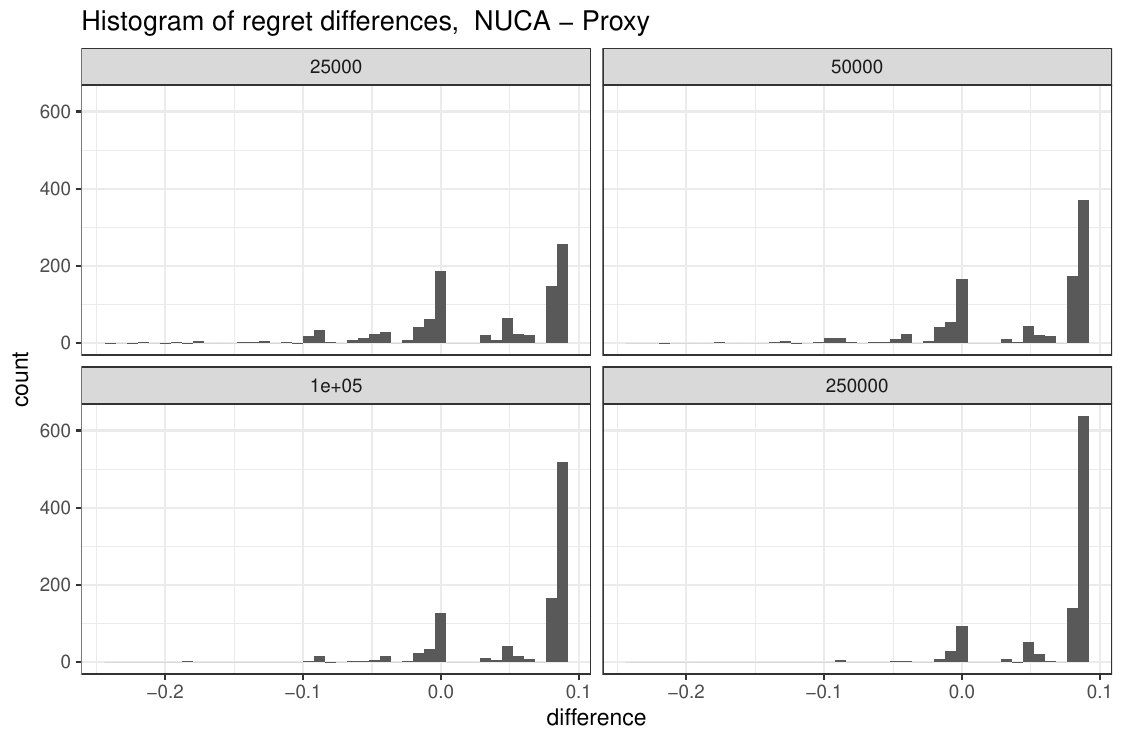}
    \caption{For each of the 4 sample sizes, we plot for each of the 1000 iterations, the difference in regret between the NUCA and proxy estimated optimal regimes. Larger difference indicates worse performance for NUCA.}
    \label{fig:regret differences}
\end{figure}

\begin{figure}
\begin{subfigure}{.5\textwidth}
  \centering
  \includegraphics[width=1\linewidth]{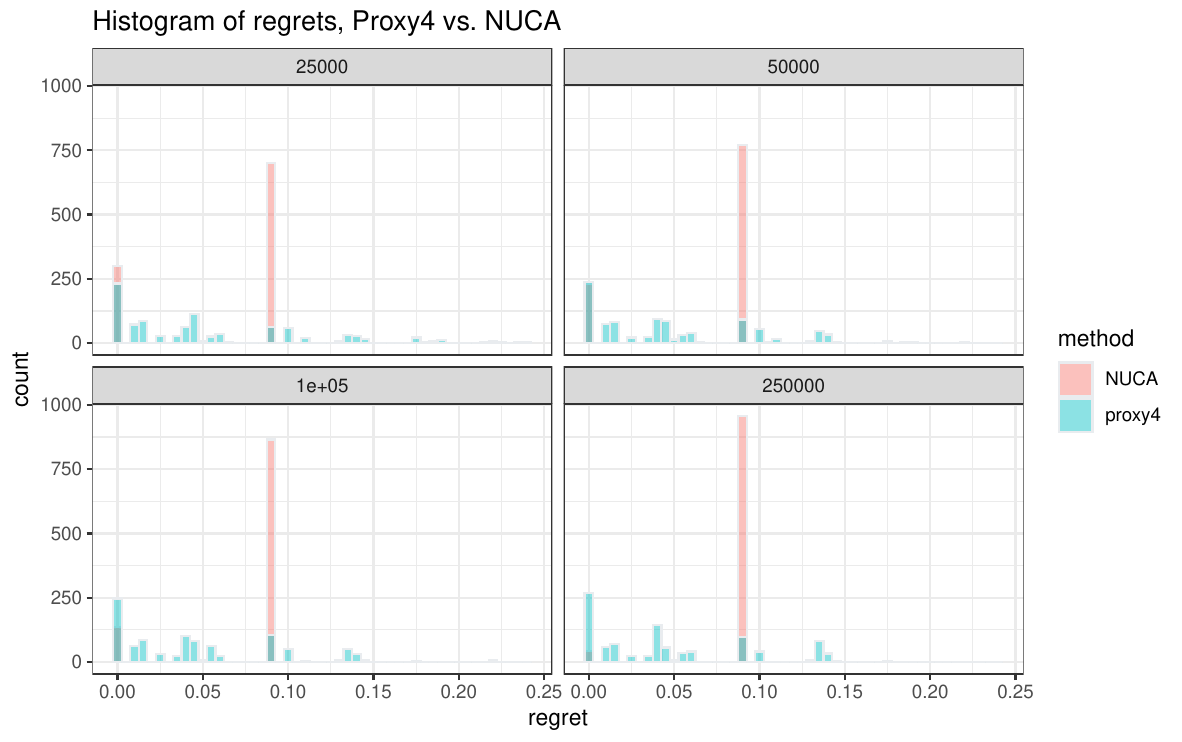}  
  \label{fig: sub4}
\end{subfigure}
\begin{subfigure}{.5\textwidth}
  \centering
  \includegraphics[width=1\linewidth]{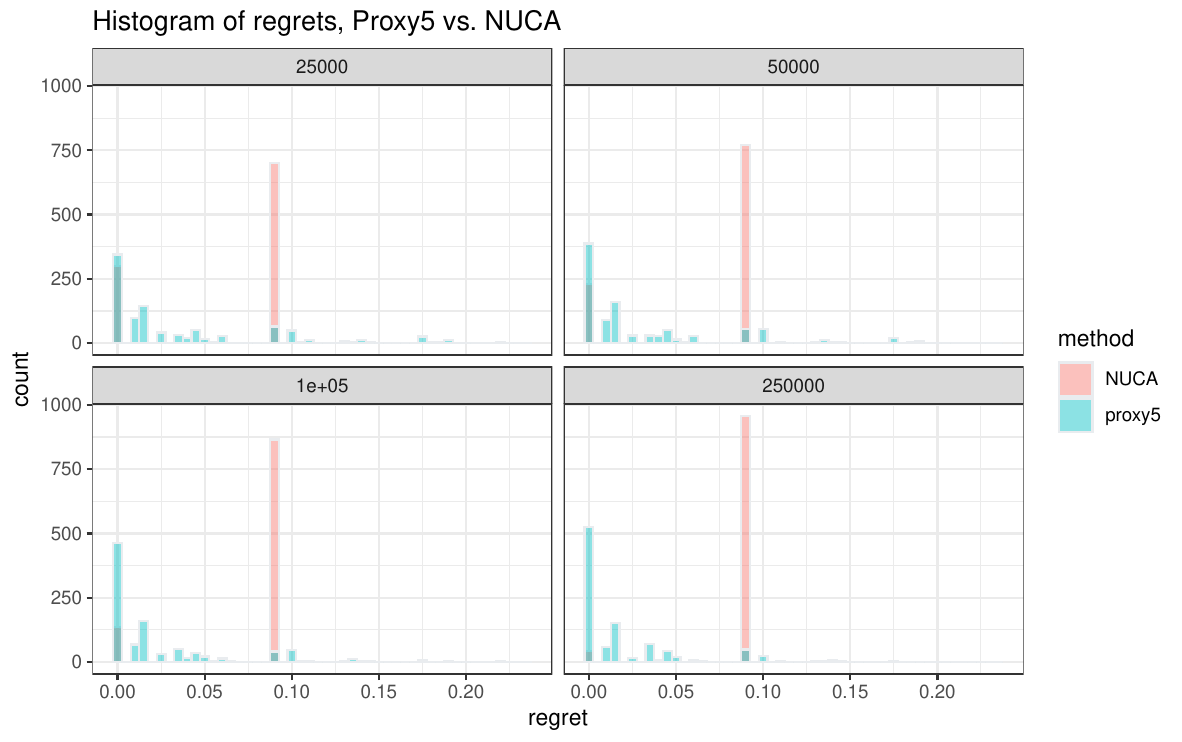}  
  \label{fig: sub5}
\end{subfigure}

\begin{subfigure}{.5\textwidth}
  \centering
  \includegraphics[width=1\linewidth]{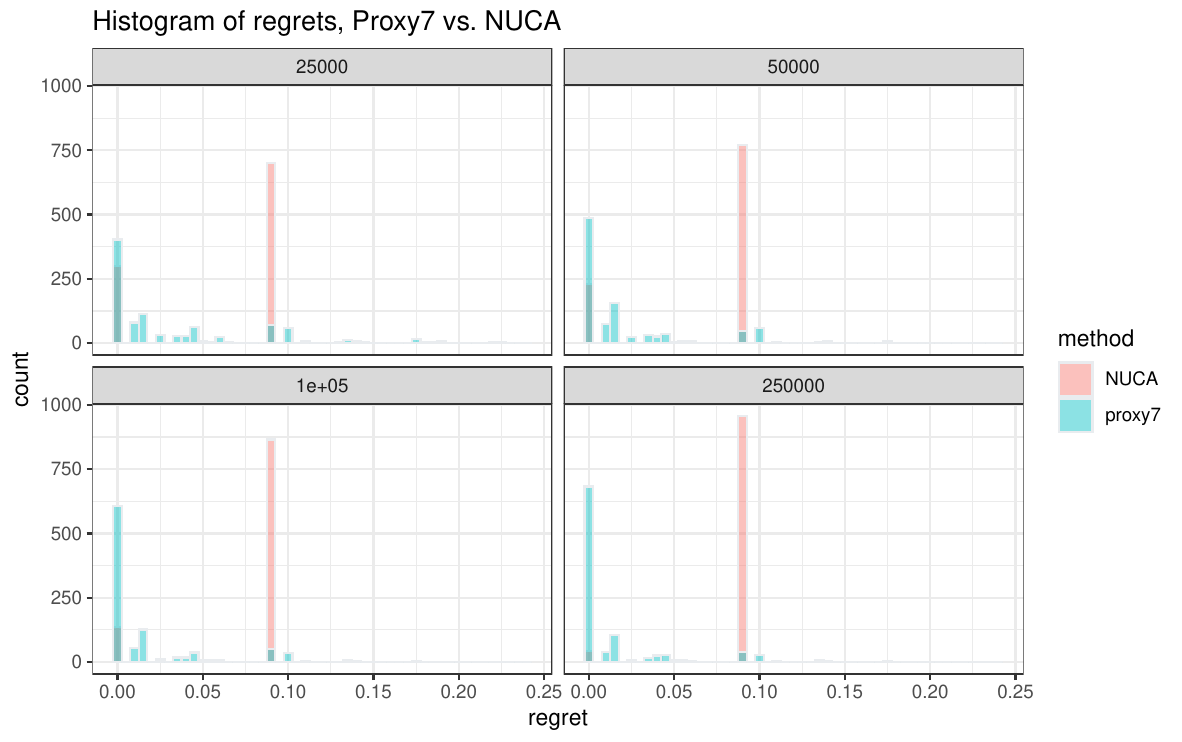}  
  \label{fig: sub7}
\end{subfigure}
\begin{subfigure}{.5\textwidth}
  \centering
  \includegraphics[width=1\linewidth]{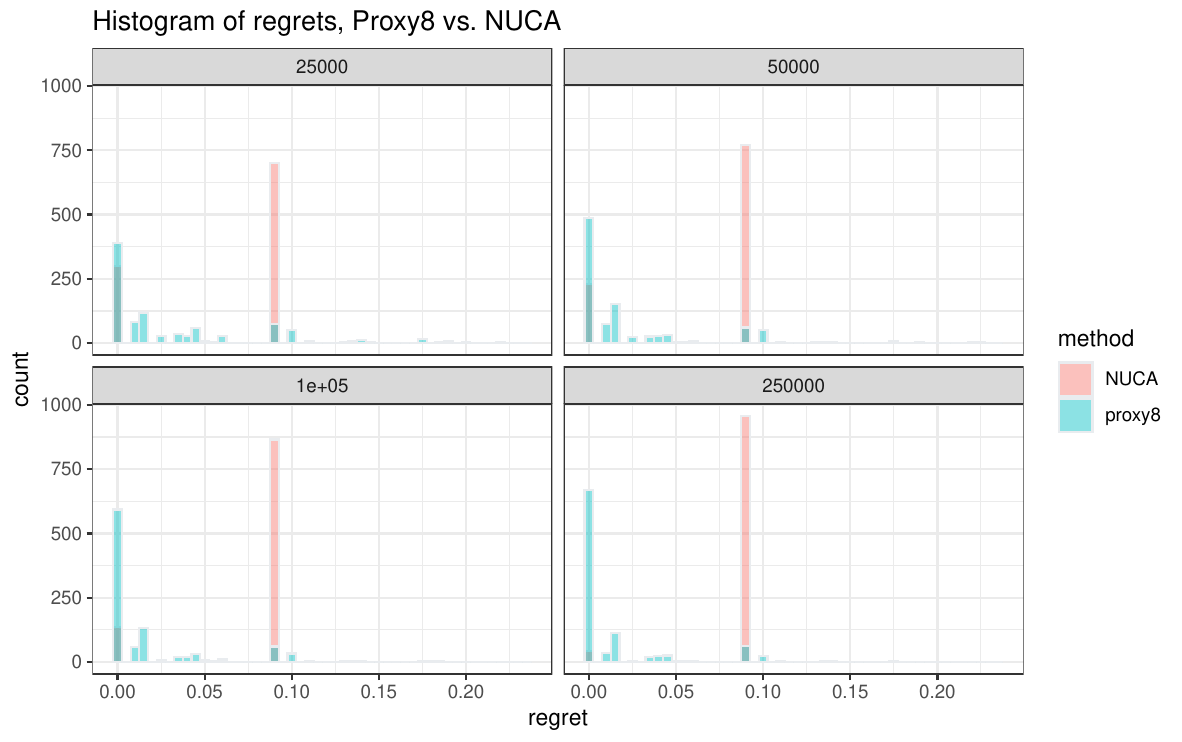}  
  \label{fig: sub8}
\end{subfigure}
\caption{Regret comparisons between the NUCA estimator and the proxy estimators fit with different orders of interactions. From clockwise, 4,5, 7, and 8.}
\label{fig: 4578}
\end{figure}

\end{document}